\documentclass[10pt, twocolumn, x11names]{article}
\usepackage[a4paper,top=2.5cm,bottom=2.5cm,left=1.675cm,right=1.675cm]{geometry}

\usepackage[T1]{fontenc}

\usepackage{fancyhdr}
\usepackage[svgnames]{xcolor}

\usepackage[inline]{enumitem}
\usepackage{multirow,booktabs}

\usepackage{times}

\usepackage[hang,small,labelfont=bf,up,textfont=it,up]{caption}
\usepackage{subcaption}
\usepackage{amsmath,amsfonts,amsthm,amssymb,dsfont}	
\setlist{nolistsep} 

\usepackage{algorithmic}
\usepackage{siunitx}
\usepackage{bm,bbm}
\usepackage{mathtools}
\mathtoolsset{showonlyrefs}
\usepackage[colorlinks=true,allcolors=blue]{hyperref}

\usepackage{graphicx}

\usepackage{numprint}
\npthousandsep{,}\npthousandthpartsep{}\npdecimalsign{.}

\newcommand{\x}{x}
\newcommand{\abs}[1]{\vert#1\vert}
\newcommand{\norm}[1]{\|#1\|}

\newcommand{\bDelta}{\boldsymbol{\Delta}_m}
\newcommand{\bDeltahat}{\hat{\boldsymbol{\Delta}}_{m,n}}
\newcommand{\tti}{2 \to \infty}
\newcommand{\Wmat}{\mathbf{W}_{m,n}}

\newcommand{\Bmat}{\mathbf{B}_m}
\newcommand{\Bmathat}{\hat{\mathbf{B}}_{m,n}}

\newcommand{\Psimat}{\mathbf{\Psi}_m}
\newcommand{\Psimathat}{\hat{\mathbf{\Psi}}_{m,n}}

\newcommand{\hatx}{\hat x_{m,n}}

\newtheorem{theorem}{Theorem}

\theoremstyle{plain}
\newtheorem{definition}{Definition}

\newtheorem{proposition}{Proposition}

\newtheorem{assumption}{Assumption}

\newcommand{\titledoc}{
Recovering manifold structure in LLM responses through a joint Euclidean mirror}
\newcommand{\titleshort}{
Recovering manifold structure in LLM responses through a joint Euclidean mirror}

\usepackage{titlesec}
\titleformat{\section}[block]{\normalfont\Large\bfseries\raggedright}{\thesection}{1em}{}
\titleformat{\subsection}[block]{\normalfont\large\bfseries\raggedright}{\thesubsection}{1em}{}

\providecommand{\keywords}[1]{{\small{\textbf{\textit{Keywords ---}} #1}}}

\usepackage{fontawesome}
\usepackage{setspace}

\usepackage{mathtools,bigints}
\mathtoolsset{showonlyrefs}

\usepackage[ruled,linesnumbered,noend]{algorithm2e}
\SetKwInput{KwInput}{Input}
\let\oldnl\nl
\newcommand{\nonl}{\renewcommand{\nl}{\let\nl\oldnl}}
\makeatother

\usepackage{authblk}
\usepackage{natbib}

\usepackage{graphicx}
\graphicspath{{Fig/}}

\allowdisplaybreaks
\usepackage{url}

\DeclareMathOperator*{\argmin}{arg\,min}

\author[1]{Maximilian~Baum}
\author[2]{Aranyak~Acharyya}
\author[3]{Tianyi~Chen}
\author[3]{Avanti~Athreya}
\author[4]{Youngser~Park}
\author[1]{Francesco~Sanna~Passino}
\author[3]{Carey~E.~Priebe}
\author[5]{Zachary~Lubberts}

\affil[1]{Department of Mathematics, Imperial College London, United Kingdom}
\affil[2]{Mathematical Institute for Data Science (MINDS),
Johns Hopkins University, United States}
\affil[3]{Department of Applied Mathematics \& Statistics, Johns Hopkins University, United States}
\affil[4]{Center for Imaging Science (CIS), Johns Hopkins University, United States}
\affil[5]{Department of Statistics, University of Virginia, United States}
\date{}

\title{\Huge\textbf{\titledoc}}

\allowdisplaybreaks

\usepackage{multibib}
\newcites{SM}{Supplementary references}

\usepackage[most]{tcolorbox}

\newtcolorbox{significancebox}{
  colback=gray!15,
  colframe=gray!60,
  boxrule=0.5pt,
  arc=2pt,
  left=6pt,
  right=6pt,
  top=6pt,
  bottom=6pt,
  width=\columnwidth,
  before skip=1em,
  after skip=1em
}

\begin{document}

\twocolumn[
\maketitle

\begin{center}
  {\small\bfseries Abstract}
\end{center}

\begin{center}
\begin{minipage}{0.85\textwidth} 
\small
Understanding the behavior of black-box large language models and determining effective means of comparing their performance is a key task in modern machine learning. We consider how large language models respond to a specific query by analyzing how the distributions of responses vary over different values of tuning parameters. 
We frame this problem in a general mathematical setting, treating the mapping from model parameters to response distributions as a structured family of probability measures, endowed with a geometry via a dissimilarity measure. We show how dissimilarities between response distributions can be represented in low-dimensional Euclidean space through a {\em joint Euclidean mirror surface} encoding the underlying geometry, which permits both qualitative and quantitative analysis of large language models and provides insight into predicting response distributions for different values of tuning parameters. We propose an estimation procedure for the underlying joint Euclidean mirror based on observed samples from the response distributions, and we prove its asymptotic properties. Additionally, we propose a statistically consistent procedure to infer the value of an unknown model parameter based on samples from the corresponding response distribution and the estimated joint Euclidean mirror. In an experimental setting with large language models, we find that changes in different tuning parameter values correspond to distinct directions in the embedding space, making it possible to estimate the tuning parameters that were used to generate a given response. 

\vspace*{0.5em}
\end{minipage}
\end{center}


\vspace*{2.5em}
]



\section{Introduction}

There is an increasing appetite for personalized AI large language models (LLMs) tailored to particular users and tasks \citep[see, for example,][]{tan2024, wozniak24,zhang2025personalization}. As these models become more specialized through differences in model architecture, training methods and training data, the ability to compare these models in a structured way takes on increasing relevance \citep{Kahng25}. Due to fundamental differences in model architectures and the fact that many model weights are not open source, comparing models based on differences in weights is not possible. It is therefore more tractable and interpretable to compare models based on the responses they generate \citep[see, for example, the \textit{LLM-as-a-judge} framework;][]{Zheng23}. 

In this work, we introduce a method to quantify differences between language models via their responses and demonstrate that it is possible to recover fundamental details about a language model using the representation of such differences.  This methodology also provides a framework to understand the sensitivity of model output to changes in tuning parameters (such as, for example, temperature), architecture, inference time or access to specialized data sources.
We present this approach as a general mathematical framework in which responses from LLMs are considered as realizations from high-dimensional probability distributions, indexed by parameters, and we equip the underlying space of distributions with a dissimilarity measure which introduces a latent geometry. We then introduce the concept of a \textit{joint Euclidean mirror} to represent differences between distributions within a low-dimensional Euclidean space, and we propose a statistically consistent estimation procedure to estimate this object when samples from these probability distributions are available. The proposed procedure is visualized in Figure~\ref{fig:mirror_viz} in an example on responses from LLMs; a more detailed description about each of the steps involved in the algorithm will be given in Section~\ref{Sec:Methodology}. We further show that the proposed joint Euclidean mirror framework could be used to recover latent generative conditions from unlabelled samples using the learned representation of distributions.
To the best of our knowledge, this is the first work to formalize the problem of recovering latent generative parameters from LLM responses within a principled mathematical framework.

\begin{figure*}[t]
\centering
\includegraphics[width=\textwidth]{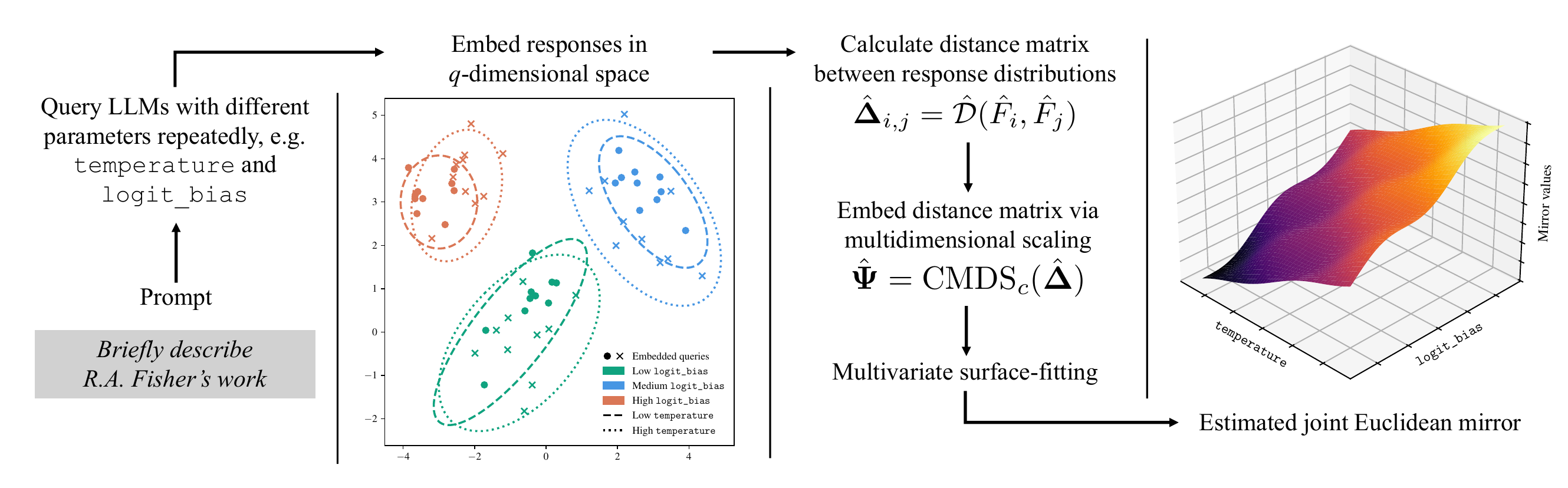}
\caption{Visualization of the joint Euclidean mirror estimation 
for 
from LLMs, for a given prompt, and two-dimensional parameters.}
\label{fig:mirror_viz}
\end{figure*}

The rest of this paper is organized as follows: Section~\ref{sec:mirror-background} discusses the required background and methodological setup for this work, followed by a description of the proposed methods in Section~\ref{Sec:Methodology}. Theoretical results about our method are proved in Section~\ref{Sec:Theory}. Section~\ref{sec:Illustrations} demonstrates the proposed algorithms and their properties on simulated data, followed by an example with responses from LLMs in Section~\ref{sec:LLM-mirrors}.

\section{Background and setup}
\label{sec:mirror-background}
Given the limited understanding of the internal functioning of LLMs, one approach to study these models is to treat them as a black-box in which one inputs a query and is returned a random response. In this way, each response from an LLM provides a single realization of a random process \citep[see, for example,][]{Helm2025}. By treating each query of an LLM as a random sample from an unknown text distribution and utilizing a text embedding which can embed any length text into a fixed $q\in\mathbb{N}$ dimensional vector, the problem of comparing different LLMs can be viewed as analogous to the problem of comparing different probability distributions in $\mathbb{R}^q$. 

The methods that we consider combine key ideas based on the concepts of the \emph{Euclidean mirror} \citep{athreya2025euclidean} as well as the study of large language models in terms of the \emph{Data Kernel Perspective Space (DKPS)} \citep{Helm2025}. The central idea of a Euclidean mirror is to define a function that maps from a non-Euclidean space into Euclidean space in such a way that a specific notion of distance between objects in the original space is preserved. In the case of dynamic graphs presented in \cite{athreya2025euclidean}, this structure is used to encode the distance between time-varying latent position distributions in terms of the maximum variation distance metric, under a random dot product graph modeling framework. In the setting proposed in \cite{athreya2025euclidean}, the evolution of latent positions for nodes in the graph is parameterized with respect to time, and the Euclidean mirror is defined as a function $\psi:[0,T] \to \mathbb{R}^c$ with $c \in \mathbb{N}$ such that for all $t, t^\prime \in [0,T]$, the Euclidean distance between $\psi(t)$ and $\psi(t^\prime)$ provides a representation for the amount of change of the latent position distributions between times $t$ and $t^\prime$. In our setting, we extend the notion of the mirror beyond the maximum variation distance metric to accommodate general notions of distance and to a multivariate setting where the distances in question are parameterized by two or more features. 


    
In this paper, the non-Euclidean space that we study is the space of distributions. Consider a set of distributions $\mathcal{F}$ parameterized by a vector $x \in \mathcal{X} \subset \mathbb{R}^d$ for $d\in\mathbb{N}$. In the case of a large language model, $F_x$ could represent the distribution of responses to a given query, embedded in $\mathbb{R}^q$, when the model is equipped with parameters $x\in\mathcal{X}$. For example, in OpenAI's ChatGPT, \texttt{temperature} and \texttt{logit\_bias} are parameters that could be varied to control the randomness of the output and specific token probabilities in the response. 

If we specify a distance metric $\mathcal{D}$ on the space $\mathcal{F}$, then the idea of a joint Euclidean mirror is to specify a function $f: \mathcal{X} \to \mathbb{R}^c$ such that distance on the space of distributions is reflected by Euclidean distance in $\mathbb{R}^c$. Depending on the set of distributions $\mathcal{F}$, the parameter space $\mathcal{X}$ and the distance metric $\mathcal{D}$ such a function may or may not exist for a given dimension of the mirror $c$. Therefore, to define the conditions for the existence of a mirror, we introduce the following notion of exact Euclidean realizability. 




\begin{definition}[Exact Euclidean realizability]
    \label{def:joint-realizable}
    Let $\mathcal{X} \subset \mathbb{R}^d$ and let $\mathcal{F}$ be a set of distributions on $\mathbb{R}^q$ such that each 
    $F_x\in\mathcal{F}$ is parameterized by some $x \in \mathcal{X}$. Let $\mathcal{D}$ be a distance metric on the space of distributions $\mathcal{F}$. 
We say that the pair $(\mathcal{F}, \mathcal{D})$ is exactly Euclidean $c$-realizable if there exists a continuous function $f:\mathcal{X} \to \mathbb{R}^c$ such that for any $x,x^\prime \in \mathcal{X}$:
    \begin{equation*}
        \norm{f(x)-f(x^\prime)} =
    \mathcal{D}(F_x, F_{x^\prime}).
\end{equation*}
Equivalently, we say that the the metric space $(\mathcal{F}, \mathcal{D})$ is exactly Euclidean $c$-realizable if 
$(\mathcal{F}, \mathcal{D})$ and $(\mathbb{R}^c, \mathcal{E})$ are isometrically equivalent, where $\mathcal{E}$ is the Euclidean distance metric.
\end{definition}

We acknowledge that assuming $(\mathcal{F}, \mathcal{D})$ and $(\mathbb{R}^c, \mathcal{E})$ are isometrically equivalent is a strong assumption, which may seem
unrealistic when $\mathcal{F}$ is the space of text‑embedding distributions with large values for $q$. 
However, 
the application to LLM responses in Section~\ref{sec:LLM-mirrors} provides empirical evidence that the assumption 
is reasonable in practice.
Based on the definition of exact Euclidean realizability, we can then provide a definition for the \textit{joint Euclidean mirror} used in this work. 

\begin{definition}[Joint Euclidean mirror]
\label{def:Mirror}
    Let $f:\mathcal{X} \to \mathbb{R}^c$ be a continuous function. We say that $f$ is a joint Euclidean mirror for the metric space $(\mathcal{F}, \mathcal{D})$ if for all $x, x^\prime \in \mathcal{X} $, \begin{equation*}
\norm{f(x)-f(x^\prime)}=\mathcal{D}(F_x,F_{x^\prime}).
\end{equation*}
\end{definition}

When a mirror exists, this structure allows us to determine which of the attributes encoded in $x$ drive significant changes in model output as measured by the selected distance metric $\mathcal{D}$ and provides us with a method to infer properties of an unobserved distribution $F_{x^{\ast}}$ corresponding to a new parameter vector $x^{\ast}$. In the case of LLMs, the parameter vector $x$ might encode features such as model temperature, access to specialized or sensitive training data or the amount of inference compute used to generate the response. Conversely, if a target distribution $F_{x^\ast}$ is observed, but the associated vector $x^\ast$ is unknown, the mirror could be used to estimate possible values of the parameter; this \emph{parameter recovery} procedure, and the conditions under which it is possible, will be made more precise in Sections~\ref{Sec:Methodology}~and~\ref{Sec:Theory}. 

The setting in \cite{Helm2025} and related literature, such as \cite{acharyya2024consistent, acharyya2025concentration}, is similar to ours: they embed large language models by defining distances between their responses to a fixed set of queries and then applying raw-stress MDS to embed the resulting distance matrix into Euclidean space, with an emphasis on providing sufficient conditions under which the estimated embedding converges to the unknown population embedding.
In contrast, we propose a more general framework that embeds the entire family of distributions, coupling the latent geometry of the space \citep[called DKPS in][]{Helm2025} encoded by the mirror, with an underlying parameter space. In this way, we can frame inference questions for LLMs as parameter inference problems.
Finally, while DKPS is based on the mean-discrepancy dissimilarity between responses, our approach allows for any metric.

\section{Methodology}
\label{Sec:Methodology}

Provided there exists a Euclidean mirror $f$ for the metric space $(\mathcal{F}, \mathcal{D})$, an immediate question of interest is how one can recover this structure given only partial information about the set $\mathcal{F}$. We study the setting where, rather than observing the the full set of distributions $\mathcal{F}$, we observe only $n\in\mathbb{N}$ independent and identically distributed (\textit{iid}) 
samples from a collection of $m\in\mathbb{N}$ distributions belonging to $\mathcal{F}$. Our problem of interest is then the construction of an estimate of the joint Euclidean mirror $f$ given only these samples. 
To do this, we propose to follow a procedure consisting of three steps: 
\begin{enumerate}
    \item Estimation of the distance matrix 
    between distributions
    for the set of parameters 
    for which samples are available;
    \item Estimation of the joint Euclidean mirror values at parameter values for which samples are observed, via classical multidimensional scaling \citep[CMDS;][]{kruskal1964multidimensional}; 
    \item Estimation of the full joint Euclidean mirror function via 
    multivariate surface fitting from the estimated values. 
\end{enumerate}
In this section, we describe each of the above steps in detail. Furthermore, in Section~\ref{Sec:Theory}, we prove theoretical results which provide a strong justification for the proposed procedure.

\paragraph{Estimation of the distance matrix.} 
Let $\x_1, \x_2, \dots, \x_m  \in \mathcal{X} \subset \mathbb{R}^d$ denote the collection of parameter vectors for which we observe samples from the corresponding distributions $F_{x_1}, \dots , F_{x_m}\in\mathcal{F}$, and let $\mathcal{S}_{x_i}$ denote an \textit{iid} sample of size $n$ from the distribution $F_{x_i}$. If we denote the empirical distribution of the sample $\mathcal{S}_{x_i}$ by $\hat F_{x_i}$, then we can calculate an estimate for the dissimilarity $\mathcal{D}(F_{x_i}, F_{x_j})$ via an estimate $\mathcal{\hat D}(\hat F_{x_i}, \hat F_{x_j})$.
Collecting these pairwise estimates yields a matrix
$\hat{\boldsymbol{\Delta}} \in \mathbb{R}^{m \times m}$ with entries
\begin{equation}
(\hat{\boldsymbol\Delta})_{i,j} = \hat{\mathcal{D}}(\hat F_{x_i}, \hat F_{x_j}), \quad i,j \in [m], \label{eq:distances_est}
\end{equation}
which serves as an empirical estimate of the pairwise distance
matrix associated with the distributions
$F_{x_1}, \ldots, F_{x_m}$.

So far, we have not specified any particular form for the dissimilarity $\mathcal{D}$, leaving it open to the requirements of the specific problem.
When the objects of interest are probability distributions, a natural and widely used choice is the Wasserstein $p$-distance, which provides a meaningful notion of dissimilarity between distributions by expressing the minimal cost of transporting probability mass from one distribution to the other. 
For 
$F_x,F_{x^\prime}\in\mathcal{F}$, the Wasserstein $p$-distance for the dissimilarity $\mathcal{D}(F_x, F_{x^\prime})$ takes the following form:  
\begin{equation}
W_p(F_x, F_{x^\prime}) =
\inf_{(Y_x,Y_{x^\prime})\sim\Gamma(F_x, F_{x^\prime})} \left(\mathbb{E}\left[ \| Y_x - Y_{x^\prime} \|^p \right] \right)^{1/p},
\end{equation}
where $p\in[1,\infty)$, and $\Gamma(F_x, F_{x^\prime})$ is the set of couplings of $F_x$ and $F_{x^\prime}$, corresponding to the set of probability measures on $\mathbb{R}^q \times \mathbb{R}^q$ whose marginals are $F_x$ and $F_{x^\prime}$.
In the discrete case, given $n$ \emph{iid} samples $y_{x,1},\dots,y_{x,n}\sim F_x$ and $y_{x^\prime,1},\dots,y_{x^\prime,n}\sim F_{x^\prime}$, the distance $\hat{\mathcal{D}}(\hat F_x, \hat F_{x^\prime})$ between the empirical measures $\hat F_x(y) = n^{-1}\sum_{i=1}^n\delta_{y_{x,i}}(y)$ and $\hat F_{x^\prime}(y) = n^{-1}\sum_{i=1}^n\delta_{y_{x^\prime,i}}(y)$ takes the form
\begin{equation}
W_p(\hat F_x, \hat F_{x^\prime}) =
\min_{\pi\in\Pi_n}\left(\frac{1}{n}\sum_{i=1}^n \|y_{x,i} - y_{x^\prime,\pi(i)}\|^p\right)^{1/p},
\end{equation}
where $\Pi_n$ is the set of all permutations of $\{1,\dots,n\}$. 

We emphasize that the methodology proposed in this work is not exclusive to this choice of distance metric or parameter space, but rather, is generally applicable to any setting where where the pair $(\mathcal{F}, \mathcal{D})$ is Euclidean $c$-realizable for some integer $c \in \mathbb{N}$. 
Alternative choices for the dissimilarity $\mathcal{D}$ between distributions include metrics such as the total variation distance, Kullback–Leibler or Jensen–Shannon divergences, and kernel-based measures like the maximum mean discrepancy (MMD). 
Earlier approaches in the literature have used simpler summaries of the difference between distributions to characterize their distance within the context of DKPS. For example, \citet{Helm2025, acharyya2024consistent, acharyya2025concentration} use $\mathcal{D}(F_x,F_{x^\prime}) = \| \mathbb{E}(Y_x) - \mathbb{E}(Y_{x^\prime})\|$ for $Y_x\sim F_x$ and $Y_{x^\prime}\sim F_{x^\prime}$. 
In contrast, the Wasserstein distance used in our work captures not only differences in location, but also differences in the spread and shape of the distributions, providing a more complete characterization of their variability. This can be particularly important for distributions expressing responses of large language models, where parameters like temperature encode randomness in the responses.

\paragraph{Mirror estimation at the observed parameters.} 
For estimation of the joint Euclidean mirror values $f(x_1),\dots,f(x_m)$ at the observed parameter values $x_1,\dots,x_m\in\mathcal{X}$, we use estimates based on the classical multidimensional scaling (CMDS) embedding technique of \citet{kruskal1964multidimensional}, a dimensionality reduction algorithm designed to embed points in a low-dimensional space in such a way that the distances between points are preserved. Let 
\begin{equation}
\hat{\boldsymbol{\Psi}} = \mathrm{CMDS}_c(\hat{\boldsymbol{\Delta}}),
\end{equation}
where $\mathrm{CMDS}_c(\cdot)$ denotes classical multidimensional scaling into $\mathbb{R}^c$, where $c$ is the joint Euclidean mirror dimension. We use each row of $\hat{\boldsymbol\Psi}$ as a mirror estimate at the observed parameter values, such that:
\begin{equation}
\tilde f(x_i) := (\hat{\boldsymbol{\Psi}})_i,\quad i\in[m].
\label{eq:mirror_estimate_parameters}
\end{equation}
At this point, we emphasize that the mirror $f$ is not unique. In particular, it is clear from Definition~\ref{def:Mirror} that applying any distance-preserving Euclidean isometry to a mirror will produce a new object that continues to satisfy the definition. Consequently, the estimates $(\hat{\boldsymbol{\Psi}})_1,\ldots,(\hat{\boldsymbol{\Psi}})_m$ produced by CMDS recover the mirror values $f(x_1),\ldots,f(x_m)$ only up to a Euclidean isometry of $\mathbb{R}^c$. In particular, we can say that $(\hat{\boldsymbol{\Psi}})_i\approx \boldsymbol{Q}f(x_i) + \boldsymbol{a}$ for $\boldsymbol{Q}\in\mathbb{O}(c)$ and $\boldsymbol{a}\in\mathbb{R}^c$. This source of non-identifiability is inconsequential in practice, because the quantities of interest depend only on pairwise distances, which are invariant under Euclidean isometries. 



\paragraph{Estimation of the joint Euclidean mirror.}
After estimating the joint Euclidean mirror values at the points corresponding to the observed parameters $x_1,\dots,x_m\in\mathcal{X}$, we now proceed to construct an estimate of the joint Euclidean mirror function for all possible values of $x\in\mathcal{X}$. This can be viewed as a multivariate surface-fitting problem: given function approximations at a finite set of points, we seek to construct a function $\hat f:\mathcal{X}\rightarrow\mathbb{R}^c$ that approximates the underlying joint Euclidean mirror function $f$ over the entire parameter space. Several classes of methods could be used for this purpose. For example, one could employ multivariate polynomial regression or spline-based surface fitting, such as tensor-product splines, which construct smooth approximations via piecewise polynomial bases \citep[see, for example,][]{deBoor62, Stone94, eilers1996flexible, Schumaker_2007}. 

Here, we propose to construct the estimated joint Euclidean mirror $\hat f$ as a piecewise linear function based on a Delaunay triangulation of the parameter values, written $\mathsf{DT}(x_1,\dots,x_m)$, which provides an optimal approach for multivariate function interpolation \citep[see, for example,][]{Chen04}. A Delaunay triangulation partitions the convex hull of the points $x_1,\dots,x_m\in\mathcal X$, written $\mathcal{X}_m = \mathsf{CH}(\{x_1,\dots,x_m\})$, into a collection of $K$ simplices $\{\mathcal{P}_1,\dots,\mathcal{P}_K\}$ with vertices taken from the observed points, 
such that $\mathcal{P}_j = \mathsf{CH}(\mathcal{V}_j)$, where $\mathcal{V}_j\subset\{x_1,\dots,x_m\}$ is a set of $d+1$ vertices, and $\mathsf{CH}$ denotes their convex hull. 
The number of simplices $K$ depends on the geometry of the points $x_1,\dots,x_m$, with $K = O(m^{\lceil d/2\rceil})$. 
For any $x\in \mathcal{X}_m$,
let $\mathcal{P}_x$ be the simplex containing $x$ from the Delaunay triangulation $\mathsf{DT}(x_1,\dots,x_m)$.
Because $\mathcal{P}_x$ is a convex set defined by 
$d+1$ vertices $[v_1(\mathcal{P}_x), \dots, v_{d+1}(\mathcal{P}_x)]\in\{x_1,\dots,x_m\}^{d+1}$, we can express $x$ using unique barycentric coordinates $\lambda_1(x), \dots, \lambda_{d+1}(x)$ satisfying 
\begin{equation}
x = \sum_{j=1}^{d+1} \lambda_j(x)\,v_j(\mathcal{P}_x), 
\end{equation}
where $\sum_{j=1}^{d+1} \lambda_j(x) =1,\ \lambda_j(x) \geq 0,\ j=1,\dots,d+1$.
Using these coordinates we build the estimate $\hat f(x)$ as the Delaunay interpolant \citep[see, for example,][]{Gillette24}:
\begin{equation}
    \hat f(x) = \sum_{j=1}^{d+1} \lambda_j(x) \tilde f\{v_j(\mathcal{P}_x)\}, \label{eq:delauney_interpolant}
\end{equation}
where $\tilde f$ is the estimated value of the mirror at the observed parameter values which we obtain from the rows of $\hat{\boldsymbol\Psi}$, \emph{cf.}~\eqref{eq:mirror_estimate_parameters}. 

The entire procedure is summarized in Algorithm~\ref{Algo:JointMirror}. 

\begin{algorithm}[t]
\caption{Joint Euclidean mirror estimation}
\label{Algo:JointMirror}

\textbf{Input} Samples $\mathcal{S}_{x_i}$ of size $n$ from 
$F_{x_i}\in\mathcal{F}$, $i \in [m]$.\\
\textbf{Output} Estimate $\hat f$ of the joint Euclidean mirror $f$. \\
\For{$r \in [m]$}{
\For{$s \in [m]$}{
\textbf{Calculate} $(\hat{\boldsymbol{\Delta}})_{r,s}=
\hat{\mathcal{D}}(\hat F_{\x_r}, \hat F_{\x_s})$, where $\hat F_{x_i}$ is the empirical distribution of $\mathcal{S}_{x_i}$.
} }

\textbf{Calculate} $\hat{\mathbf{\Psi}} = \mathrm{CMDS}_c(\hat{\boldsymbol{\Delta}}) \in \mathbb{R}^{m \times c}$, where $\mathrm{CMDS}_c(\cdot)$ corresponds to CMDS into $\mathbb{R}^c$. 

\textbf{Construct} an estimated joint Euclidean mirror $\hat{f}$ via a surface-fitting method, such as interpolating the points $\{[x_i,(\hat{\mathbf{\Psi}})_i]\; i \in[m]\}$ via the Delaunay triangulation interpolant, \emph{cf.}~Equation~\eqref{eq:delauney_interpolant}.\\
\end{algorithm}

\subsection{Parameter recovery}

We now consider the case in which, in addition to \emph{iid} samples from $F_{x_1},\dots,F_{x_m}\in\mathcal{F}$ for known parameters $\x_1, \x_2, \dots, \x_m  \in \mathcal{X}$, additional \emph{iid} samples $\mathcal{S}_{x^\ast}$ from a distribution $F_{x^\ast}\in\mathcal{F}$ are available, but the parameter $x^\ast\in\mathcal{X}$ is unknown. We denote $\mathcal{S}_{x^\ast}$ as the \emph{unlabeled} samples. Under this framework, an immediate question of interest is whether the unknown parameter $x^\ast$ can be consistently estimated from the unlabeled samples $\mathcal{S}_{x^\ast}$, leveraging the information contained in the labeled samples $\{\mathcal{S}_{x_1}, \dots, \mathcal{S}_{x_m}\}$ and the structure 
encoded by the underlying joint Euclidean mirror $f$. We call this task \emph{parameter recovery}. 

To estimate the value of the underlying parameter for the unlabeled samples, we propose an adjustment to Algorithm~\ref{Algo:JointMirror}. In particular, we propose to calculate the $(m+1)\times(m+1)$ distance matrix from the samples $\{\mathcal{S}_{x_1},\dots,\mathcal{S}_{x_m},\mathcal{S}_{x^\ast}\}$, with entries consisting of the pairwise estimated dissimilarities, as in \eqref{eq:distances_est}. Next, classical multidimensional scaling is applied to the estimated distance matrix to obtain a matrix $\hat{\boldsymbol{\Psi}}\in\mathbb{R}^{(m+1)\times c}$, and the estimate of the mirror $\hat f$ is constructed by multivariate surface-fitting based only on the first $m$ rows of the matrix, for which corresponding observed parameter values $x_1,\dots,x_m$ are available. The estimate $\hat x$ of the unknown parameter $x^\ast$ of the distribution corresponding to unlabeled samples is then obtained by minimizing the difference between the mirror value for the unknown parameter, corresponding to $(\hat{\boldsymbol\Psi})_{m+1}$, and the estimated joint Euclidean mirror $\hat f$, as follows: 
\begin{equation}
    \hat x = \argmin_{s\in\mathcal{X}_m}\norm{\hat f(s)- (\hat{\boldsymbol{\Psi}})_{m+1}}.
\end{equation}
In Section~\ref{Sec:Theory}, we will prove that $\hat x$ provides a consistent estimate for $x^\ast$. The procedure is summarized in Algorithm~\ref{alg:Recovery}.

\begin{algorithm}[t]
    \caption{Parameter recovery via the joint mirror}
    \label{alg:Recovery}
    \textbf{Input} Samples $\mathcal{S}_{x_i}$ for known $x_1,\dots,x_m \in \mathcal{X}$; unlabeled sample $\mathcal{S}_{x_{m+1}}$ for $x_{m+1}\in\mathcal{X}$ unknown. \\ 
    \textbf{Output} Estimate $\hat{x}$ for the unknown parameters $x_{m+1}$. \\ 
    \For{$r \in [m+1]$}{
    \For{$s \in [m+1]$}{
    \textbf{Calculate} $(\hat{\boldsymbol{\Delta}})_{r,s}=
    \hat{\mathcal{D}}(\hat F_{\x_r}, \hat F_{\x_s})$, where $\hat F_{x_i}$ is the empirical distribution of $\mathcal{S}_{x_i}$.
    } }
    \textbf{Calculate} $\hat{\mathbf{\Psi}} = \mathrm{CMDS}_c(\hat{\boldsymbol{\Delta}}) \in \mathbb{R}^{(m+1) \times c}$, where $\mathrm{CMDS}_c(\cdot)$ corresponds to CMDS into $\mathbb{R}^c$. \\
    \textbf{Construct} an estimated mirror $\hat f(x)$ by interpolating the first $m$ rows of $\hat{\mathbf{\Psi}}$ corresponding to the labeled samples (using, for example, Delaunay triangulation). \\
    \textbf{Define} the mirror value $\hat{f}^\ast$ as the $(m+1)$-th row of $\hat{\mathbf{\Psi}}$, corresponding to the unlabeled sample. \\
    \textbf{Return} the estimated parameter estimate 
    \begin{equation}
    \hat x = \argmin_{s\in\mathcal{X}_m}\norm{\hat f(s)- \hat f^\ast}.
    \end{equation}
\end{algorithm}

This approach to parameter recovery is particularly appealing for practical applications, as it provides a statistically principled data-driven approach to recovering unknown latent parameters from unlabeled samples. In the context of large language models, depending on the choice of the parameter space $\mathcal{X}$, this framework could, for example, enable the identification of implicit training or tuning characteristics, or potentially reveal access to sensitive or proprietary information encoded in the generative process. Such a perspective appears largely unexplored in the literature, and could potentially open new directions for inference in complex generative models.

\section{Theoretical results}
\label{Sec:Theory}

In this section, we prove theoretical results related to the steps detailed in Algorithms~\ref{Algo:JointMirror} and~\ref{alg:Recovery}, demonstrating that the proposed procedure is mathematically principled and 
consistent. 

Consider the framework detailed in Section~\ref{Sec:Methodology}, in which \emph{iid} samples of size $n$ are observed from distributions $F_{x_1}, \dots , F_{x_m}\in\mathcal{F}$ for $\x_1, \x_2, \dots, \x_m  \in \mathcal{X} \subset \mathbb{R}^d$, and Algorithm~\ref{Algo:JointMirror} is used to obtain an estimate $\hat f$ of the joint Euclidean mirror. In this section, we derive probabilistic bounds showing that, if the selected estimator $\mathcal{\hat D}(\hat F_{x_i}, \hat F_{x_j})$ of the dissimilarity measure provides a good approximation of the distance metric $\mathcal{D}$, then it is possible to show that the estimated mirror $\hat f$ converges to a valid mirror satisfying Definition~\ref{def:Mirror} in the asymptotic regime where both the number of observed parameter values $m$ and number of samples $n$ increase. 
For our probabilistic bounds, we use the concept of \emph{overwhelming probability} \citep[see, for example,][]{tao2010random}. An event $E_n$ depending on a parameter $n$ holds with overwhelming probability if, for every fixed $\alpha>0$, there exists a constant $C_\alpha>0$ independent of $n$ such that $\mathbb{P}(E) \geq 1 - C_\alpha n^{-\alpha}$ holds.

To make the notion of a good estimator for the distances more precise, we define the theoretical matrix of pairwise distances $\bDelta \in \mathbb{R}^{m \times m}$ and the empirical matrix $\bDeltahat \in \mathbb{R}^{m \times m}$ of pairwise distances for the 
$m$ observed parameters as:
\begin{align}
(\bDelta)_{i,j}= \mathcal{D}(F_{\x_i}, F_{\x_j}), & & (\bDeltahat)_{i,j}= \mathcal{\hat D}(\hat F_{\x_i}^n, \hat F_{\x_j}^n),
\end{align}
where the 
distributions $\hat F_{x}^n$ are 
estimated from $n$ \emph{iid} samples from the corresponding distribution $F_{x}.   $
To ensure that the estimated mirror $\hat f$ converges to a valid joint Euclidean mirror, it is required that the true distances $\mathcal{D}(F_{x}, F_{x^\prime})$ are well-represented by the finite-sample counterpart $\mathcal{\hat D}(\hat F_{x}, \hat F_{x^\prime})$ such that $\norm{\bDelta  - \bDeltahat}_F \rightarrow 0$ for $m,n\to\infty$ with overwhelming probability. Because the dimensionality of $\bDelta$ grows with $m$, 
it is necessary to ensure that the sample size $n$ grows sufficiently quickly relative to the number of parameters $m$. The necessary growth rate of $n$ relative to $m$ will depend on the specific choice of $\mathcal{D}$ and estimator $\hat{\mathcal{D}}$. Under Assumption~\ref{cond:m-n-growth} and a finite moment assumption on the distributions, Proposition~\ref{prop:consistent-delta-estimation} demonstrates that this condition can be satisfied for the Wasserstein 1-distance proposed in Section~\ref{Sec:Methodology} and used in the examples in Sections \ref{sec:normal-mirrors} and \ref{sec:LLM-mirrors}. 

\begin{assumption} \label{cond:m-n-growth}
    Suppose the number of distributions $m$ depends on the sample size $n$, and denote this by $m(n)$. Let $q > 2$ be a fixed integer denoting the dimensionality of each distribution. We assume
    \begin{equation*}
    \begin{aligned}
    \lim_{n \to \infty}
        m(n)
        \left(
\frac{\mathrm{log}^2(n)}{n}
        \right)^{1/q} 
        = 0.
    \end{aligned}
    \end{equation*}
    Note that this is satisfied when $n = \Omega(m^{^{q+1}})$. 
\end{assumption}

Assumption~\ref{cond:m-n-growth}
ensures that the number of samples from every response distribution grows sufficiently quickly relative to the number of distributions under consideration. 
Under this regime, Proposition~\ref{prop:consistent-delta-estimation} establishes an upper bound on the rate at which the estimated distances between the  response distributions approach their population counterparts.

\begin{proposition}
    \label{prop:consistent-delta-estimation}
    Let $\mathcal{F}$ be a set of distributions on $\mathbb{R}^q$. Suppose there exists a $\gamma>0$ and a constant $C_{\mathcal{F}}$ such that for each distribution $F \in \mathcal{F}$, the moment condition $\int_{\mathbb{R}^q} e^{\gamma \vert s \vert^2} dF(s) \leq C_{\mathcal{F}}$ holds. Let $\mathcal{D}$ and $\hat{\mathcal{D}}$ be the Wasserstein 1-distance.
Suppose $m(n)$ satisfies Assumption \ref{cond:m-n-growth}. 
Then, there exists $n^\ast\in\mathbb{N}$ such that, if $n>n^\ast$:
\begin{equation}
\| \bDeltahat - \bDelta \|_F < m\left(\frac{\log^2(n)}{n}\right)^{1/q}
\end{equation}
holds with overwhelming probability. 
\end{proposition}
\begin{proof}
The result is proved in Section~\ref{proof:consistent-delta-estimation}.
\end{proof}

This proposition derives a bound for the convergence of the estimated Wasserstein distances and their theoretical counterparts. 
It must be remarked that the bound heavily penalizes the dimensionality of the distribution. However, this appears to be close to the best possible result when using the Wasserstein 1-distance for continuous distributions on $\mathbb{R}^q$ as the expected value of the error scales as $n^{-1/q}$. For details see the note on the curse of dimensionality in \cite{panaretos2019statistical}, Section 3.3. This convergence rate can be improved under additional assumptions on the class of distributions under consideration \citep[for example, smoothness of the density; see][Theorem 2.18]{chewi2025statistical}. Empirically, we often observe relatively small errors without requiring extremely large sample sizes, which suggests that such additional structure often holds in practice (\emph{cf.} Section~\ref{sec:LLM-mirrors}). Furthermore, if Wasserstein-based rates are prohibitively slow, one can instead use alternative metrics on distribution spaces (for example, energy-based distances) as discussed in Section~\ref{Sec:Methodology}.

Using the convergence guaranteed by Proposition~\ref{prop:consistent-delta-estimation}, we can further show that, for $m$ and $n$ tending to infinity, the estimated mirror converges to a true mirror for all \textit{observed} values of $x_i$ for $i\in[m]$. To formalize this concept, we introduce the matrix $\Psimat \in \mathbb{R}^{m \times c}$ which represents a discrete analogue of the mirror $f$ and encodes the distances between the each of the observed distributions $F_{x_i}$ for $i \in [m]$. Concretely, if we use $(\Psimat)_i$ to denote the $i$-th row of $\Psimat$, then the matrix $\Psimat$ satisfies $\norm{(\Psimat)_i - (\Psimat)_j} = \mathcal{D}(F_{x_i}, F_{x_j})$ for all $i,j \in [m]$. The existence of such a matrix is guaranteed by the the realizability of the set $(\mathcal{F}, \mathcal{D})$. The formal statement of this result is expressed in Theorem~\ref{conj:psi_convergence_growing_m_and_n}, under the technical assumptions in Assumptions~\ref{assum:boundedset} and \ref{assum:smallest-eigenvalue}.

\begin{assumption} \label{assum:boundedset}
Suppose that $\mathcal{F}$ is bounded with respect to the distance $\mathcal{D}$, so that for any $f_1,f_2\in\mathcal{F}$, there exists a finite constant $D_{\mathrm{max}}$ such that:
\begin{equation}
\mathcal{D}(f_1,f_2)\leq D_{\mathrm{max}}.
\end{equation}
\end{assumption}

\begin{assumption}\label{assum:smallest-eigenvalue}
    Let $\Bmat = -\frac{1}{2} \mathbf{H}_{m} \bDelta^{\odot 2} \mathbf{H}_m$ where $\bDelta^{\odot 2}$ denotes the element-wise square of $\bDelta$ and $\mathbf{H}_m$ is the centering matrix of size $m$. We assume that there exists constants $C >0$ and $m^\ast \in \mathbb{N}$ such that $\lambda_c(\Bmat) \geq C m$ holds for all $m > m^\ast$.
\end{assumption}


\begin{theorem}
\label{conj:psi_convergence_growing_m_and_n}
Consider $\Psimat = \mathrm{CMDS}_c(\bDelta)$ and $\Psimathat = \mathrm{CMDS}_c(\bDeltahat)$, and let $\mathcal{D}$ and $\hat{\mathcal{D}}$ be the Wasserstein 1-distance. 
Under 
Assumptions~\ref{cond:m-n-growth}--\ref{assum:smallest-eigenvalue}, 
there exist a constant $K>0$ and 
a sequence of orthogonal matrices $\Wmat \in \mathbb{O}(c)$ such that for any $n$ sufficiently large,
\begin{equation*}
\norm{\Psimat - \Psimathat \Wmat}_{\tti} \leq K 
m^{1/2} \left(\frac{\log^2(n)}{n}\right)^{1/q}
\end{equation*}
holds with overwhelming probability.
\end{theorem}
\begin{proof}
    The result is proved in Section~\ref{proof:psi_convergence_growing_m_and_n}.
\end{proof}

The orthogonal matrix $\Wmat$ accounts for the non-identifiability of the joint Euclidean mirror (\textit{cf.} Section~\ref{Sec:Methodology}). 
Given Theorem~\ref{conj:psi_convergence_growing_m_and_n}, which shows the uniform convergence of the estimated mirror to the true mirror at all observed points $x_i$, it remains to show the convergence of the estimated mirror $\hat f$ to a joint Euclidean mirror $f$ over the full space. If the estimated mirror $\hat f$ is constructed via the Delaunay interpolant as described in Section~\ref{Sec:Methodology}, the target joint Euclidean mirror $f$ is sufficiently smooth, and the parameter space is bounded and sufficiently well-sampled (as supposed in Assumptions~\ref{ass:lipschitz_mirror}, \ref{ass:bounded} and~\ref{cond:dense-m} respectively), then the convergence is established by Theorem~\ref{conj:joint-mirror} below. 

\begin{assumption} 
\label{ass:lipschitz_mirror}
Suppose that the joint Euclidean mirror $f$ is $C$-Lipschitz continuous for some constant $C>0$, such that for all $x,x^\prime\in\mathcal{X}$, we have:
\begin{equation}
\norm{f(x) - f(x^\prime)} \le C \norm{x - x^\prime}.
\end{equation}
\end{assumption}

\begin{assumption} \label{ass:bounded}
Suppose that $\mathcal{X}\subset\mathbb{R}^d$ is bounded.
\end{assumption}


\begin{assumption}\label{cond:dense-m}
    Suppose that the observed parameter values $x_1, \dots, x_m$ densely cover the parameter space $\mathcal{X}$ such that as $m$ increases, the convex hull $\mathcal{X}_m=\mathsf{CH}\{x_1, \dots, x_m\}$ converges to $\mathcal{X}$ and the maximum diameter of any simplex in the Delaunay triangulation of $x_1, \dots, x_m$ converges to zero. 
\end{assumption}


\begin{theorem}
\label{conj:joint-mirror}
    Let $\mathcal{X}_m=\mathsf{CH}\{x_1, \dots, x_m\}$ and let $\hat f_{m,n}$ denote the estimated mirror produced by Algorithm~\ref{Algo:JointMirror} using the Delaunay linear interpolant. If the pair $(\mathcal{F}, \mathcal{D})$ is exactly Euclidean $c$-realizable and Assumptions~\ref{cond:m-n-growth}--\ref{cond:dense-m} 
    are satisfied, then there exists a sequence of mirrors $f_{m,n}$ each satisfying Definition~\ref{def:Mirror}, such that for any $\varepsilon>0$:
    \begin{equation}
    \mathbb{P}\left\{\sup_{x\in \mathcal{X}_m} \abs{\hat f_{m,n}(x) - f_{m,n}(x)} > \varepsilon \right\} \rightarrow 0.
    \end{equation}
\end{theorem}

\begin{proof}
The result is proved in Section~\ref{proof:joint-mirror}.
\end{proof}

It remains to be shown that the parameter recovery procedure based on $\hat f$ described in Algorithm~\ref{alg:Recovery} yields consistent estimates of the true parameter $x^\ast \in \mathcal{X}_m$, when the unlabeled samples are drawn from $F_{x^\ast} \in \mathcal{F}$.
This consistency is established by Theorem~\ref{conj:consistent_param_estim}. 

\begin{theorem}
\label{conj:consistent_param_estim}
    Suppose Assumptions \ref{cond:m-n-growth}--\ref{cond:dense-m} hold. Let $\hatx$ be the estimated parameter value produced by Algorithm \ref{alg:Recovery} for a collection of responses sampled from 
    $F_{x^\ast}\in\mathcal{F}$, for a fixed but unknown parameter value $x^\ast \in \mathcal{X}_m$. Suppose that the dimension of the mirror is selected to be equal to that of the parameter space $\mathcal{X}$, such that $c=d$. If the joint Euclidean mirror 
    $f:\mathbb{R}^d \to \mathbb{R}^d$ is invertible with Jacobian matrix $\mathcal{J}_f(x) \in \mathbb{R}^{d \times d}$,
    such that the smallest singular value 
    satisfies  $\sigma_{\mathrm{min}}\{\mathcal{J}_f(z)\} \geq a$ for some $a > 0$ 
    for all $z \in \mathcal{X}$,
    then for any $\varepsilon >0$: 
    \begin{equation}
    \mathbb{P}\left\{\norm{x^\ast - \hat x_{m,n}} > \varepsilon \right\} \rightarrow 0.
    \end{equation}
\end{theorem}

\begin{proof}
    The result is proved in Section~\ref{proof:consistent_param_estim}. 
\end{proof}

All theoretical results are visually summarized in Figure~\ref{fig:diagram}, which expresses the entire estimation and recovery procedure proposed in Algorithms~\ref{Algo:JointMirror} and~\ref{alg:Recovery}, and its asymptotic properties. 

\begin{figure}[t]
\centering
\includegraphics[width=0.475\textwidth]{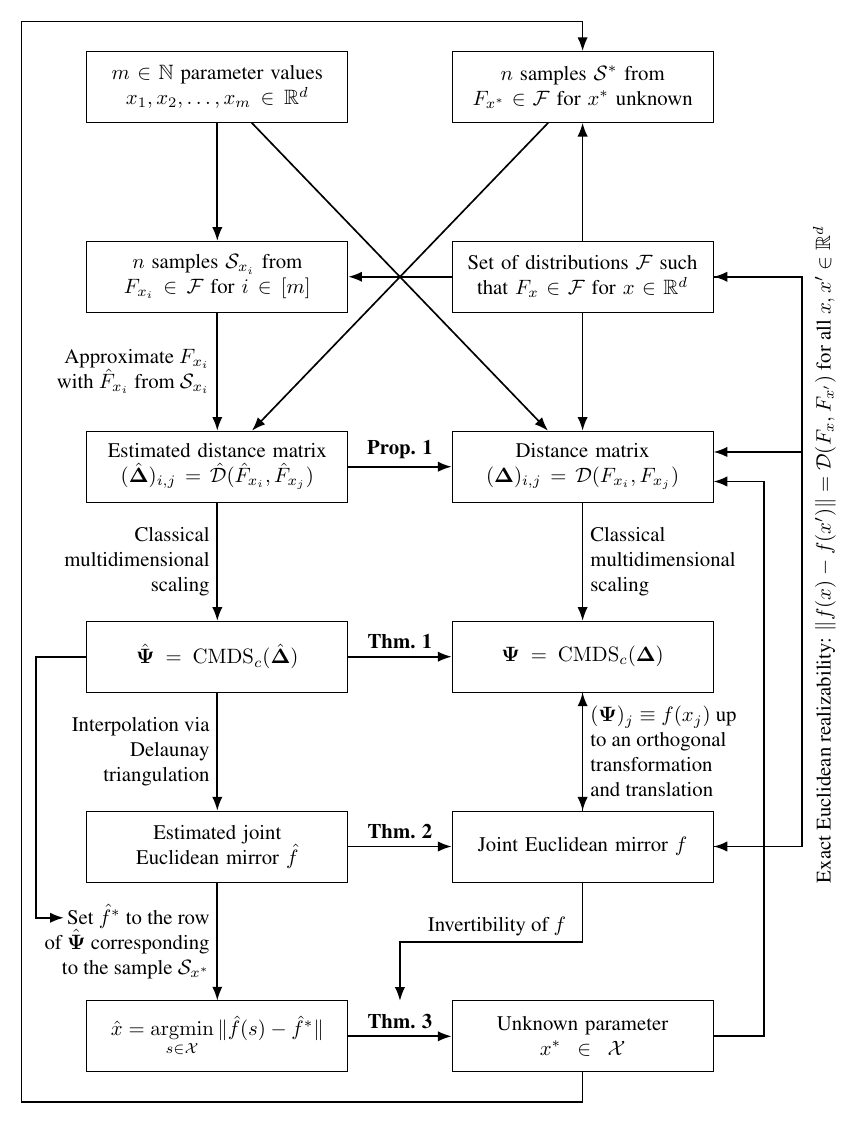}
\caption{Summary diagram of asymptotic theoretical properties of the joint Euclidean mirror estimation and parameter recovery procedures proposed in Algorithms~\ref{Algo:JointMirror} and~\ref{alg:Recovery}.}
\label{fig:diagram}
\end{figure}



\section{Illustrative examples}
\label{sec:Illustrations}

\subsection{Mirror estimation}
\label{sec:normal-mirrors}

To illustrate the application of Algorithm \ref{Algo:JointMirror}, we present an example in which a mirror $f$ exists and is known. Let $\mathcal{X} \subset \mathbb{R}^2$ be the $[1,10] \times [1,10]$ plane and define $\mathcal{F}$ to be the set of normal distributions parameterized by $x$ such that $F_x \sim \mathcal{N}(\mu_x, 1)$, where $\mu_x = 0.1 \, \norm{x - (5.5,5.5)}^2$. In this example, we set $m=100$ and take $\x_1, \dots, \x_m \in \mathcal{X}$ to be the set of points on the integer values of the $[1,10]\times [1,10]$ grid. For two normal distributions with equal variance, the Wasserstein 1-distance between the distributions is equal to the absolute value of the difference between their location parameters, resulting in 
$W_1(F_{x},F_{x^\prime})=|\mu_x - \mu_{x^\prime}|$. Therefore, if we take $\mathcal{D}$ to be the Wasserstein 1-distance, it follows that the pair $(\mathcal{F}, \mathcal{D})$ is Euclidean 1-realizable with mirror  $f:\mathbb{R}^2\rightarrow\mathbb{R}$ given by $f(x) = 0.1 \,\norm{x - (5.5,5.5)}^2$. 

According to the theoretical results in Section~\ref{Sec:Theory}, we expect that, by observing a sample of size $n$ from each of the $m$ distributions, Algorithm~\ref{Algo:JointMirror} will produce a surface that converges to a valid mirror. More specifically, the estimated mirror function $\hat f_{m,n}(x)$ will converge to a rigid transformation of the function $f(x)$ as the number of observations $n$ and parameter values $m$ increase. In Figure~\ref{fig:normal-surfaces} we generate samples from these distributions and present the estimated surfaces from the experimental setting where $m=100$ is fixed and $n$ increases. The resulting error in recovering the joint Euclidean mirror at the locations of the observed parameter values is plotted in Figure~\ref{fig:simulation-error}. 
By increasing $n$ from 10 to 500 we see that the estimated mirror $\hat f_{m,n}$ takes on the parabolic shape of the true mirror $f$. Increasing the value of $n$ leads to increased accuracy of the mirror $\hat f_{m,n}(x)$ at each of the $m$ points that we observe, whereas increasing the value of $m$ would produce figures where we obtain estimates of the mirror $f(x)$ evaluated on a finer mesh of points. 



\begin{figure*}
     \centering
     \begin{subfigure}[b]{0.24\textwidth}
         \centering
         \caption{$n=10$}
         \label{fig:normal-surface-10}
         \includegraphics[width=\textwidth]{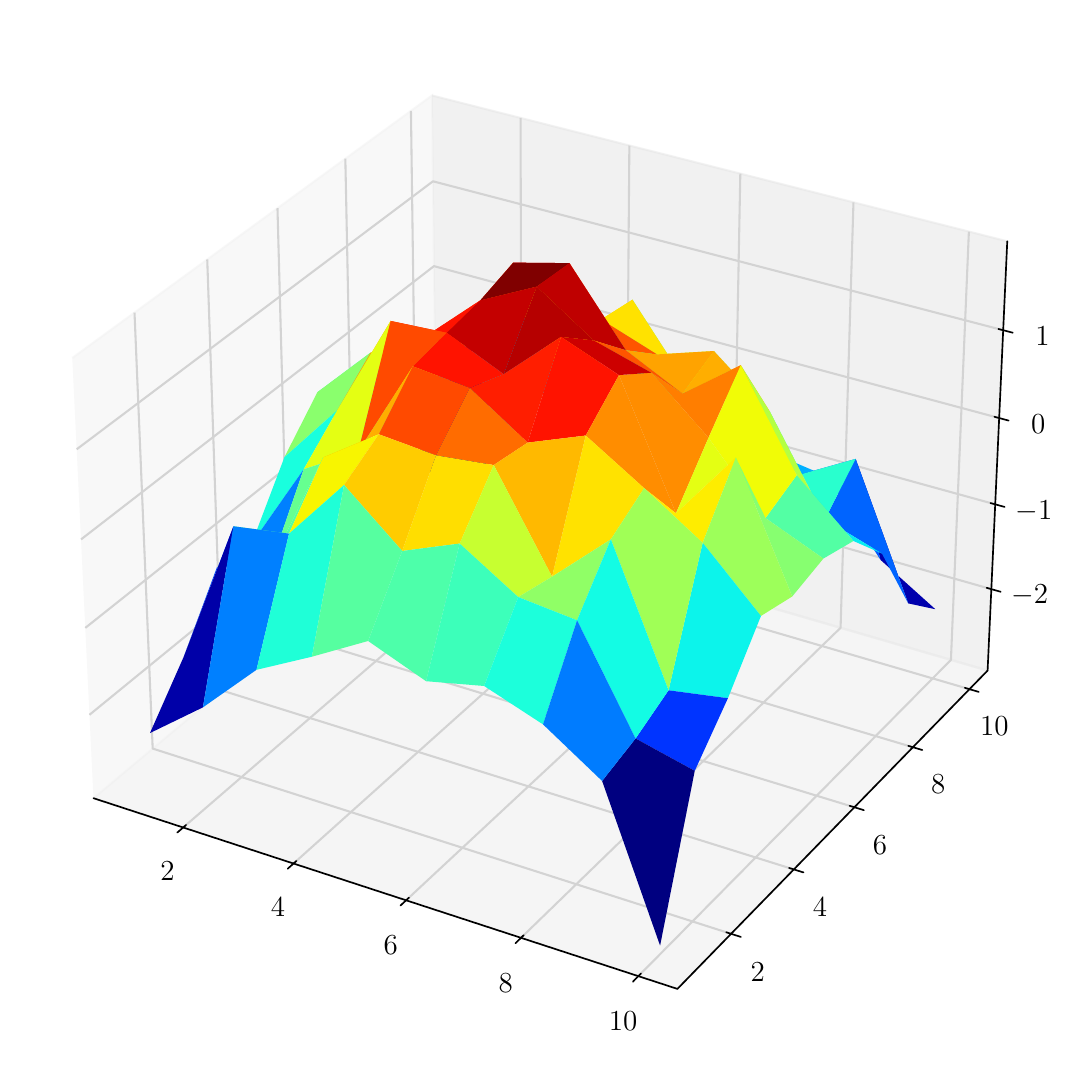}
     \end{subfigure}
     \hfill
     \begin{subfigure}[b]{0.24\textwidth}
         \centering
         \caption{$n=50$}
         \label{fig:normal-surface-50}
         \includegraphics[width=\textwidth]{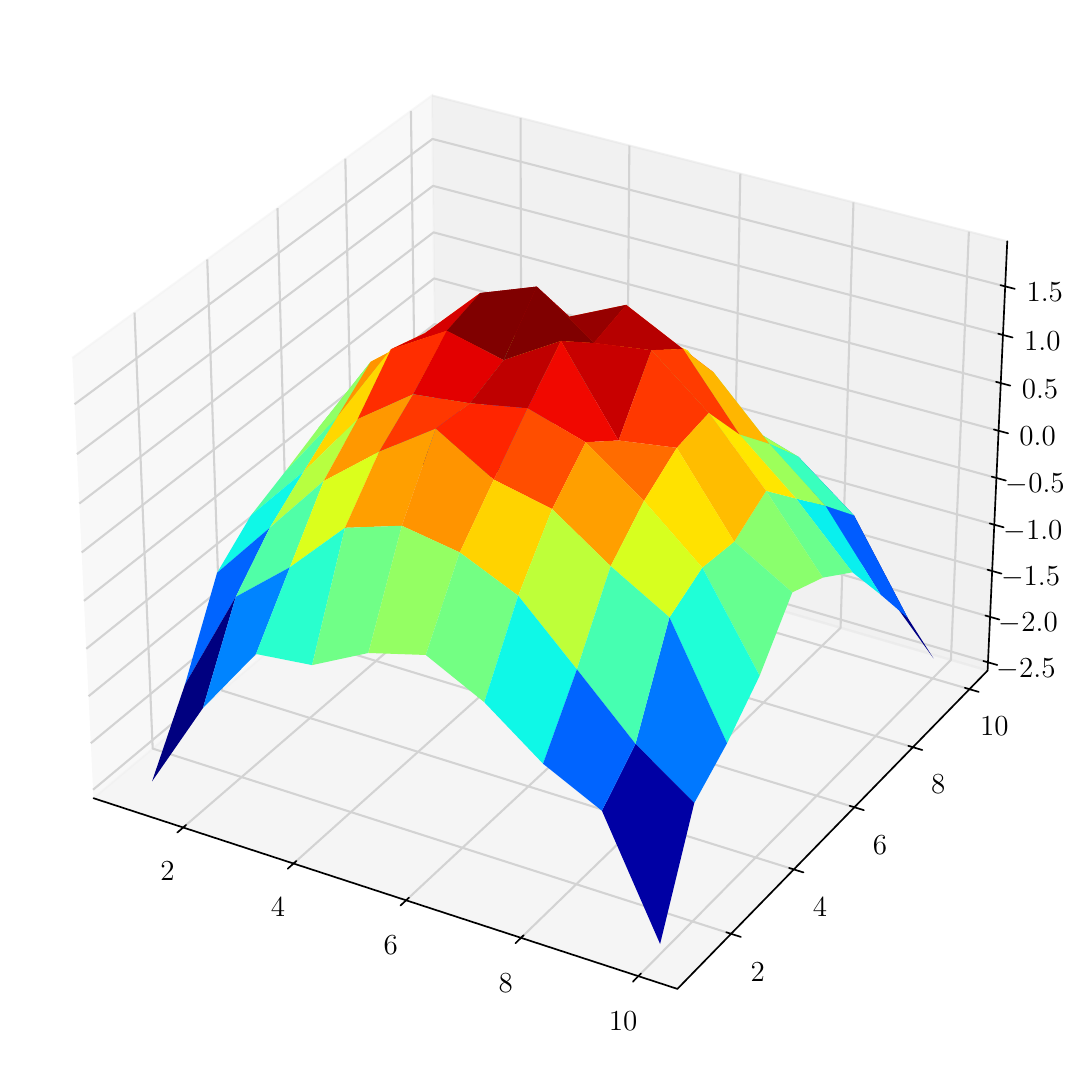}
     \end{subfigure}
     \begin{subfigure}[b]{0.24\textwidth}
         \centering
         \caption{$n=100$}
         \label{fig:normal-surface-100}
         \includegraphics[width=\textwidth]{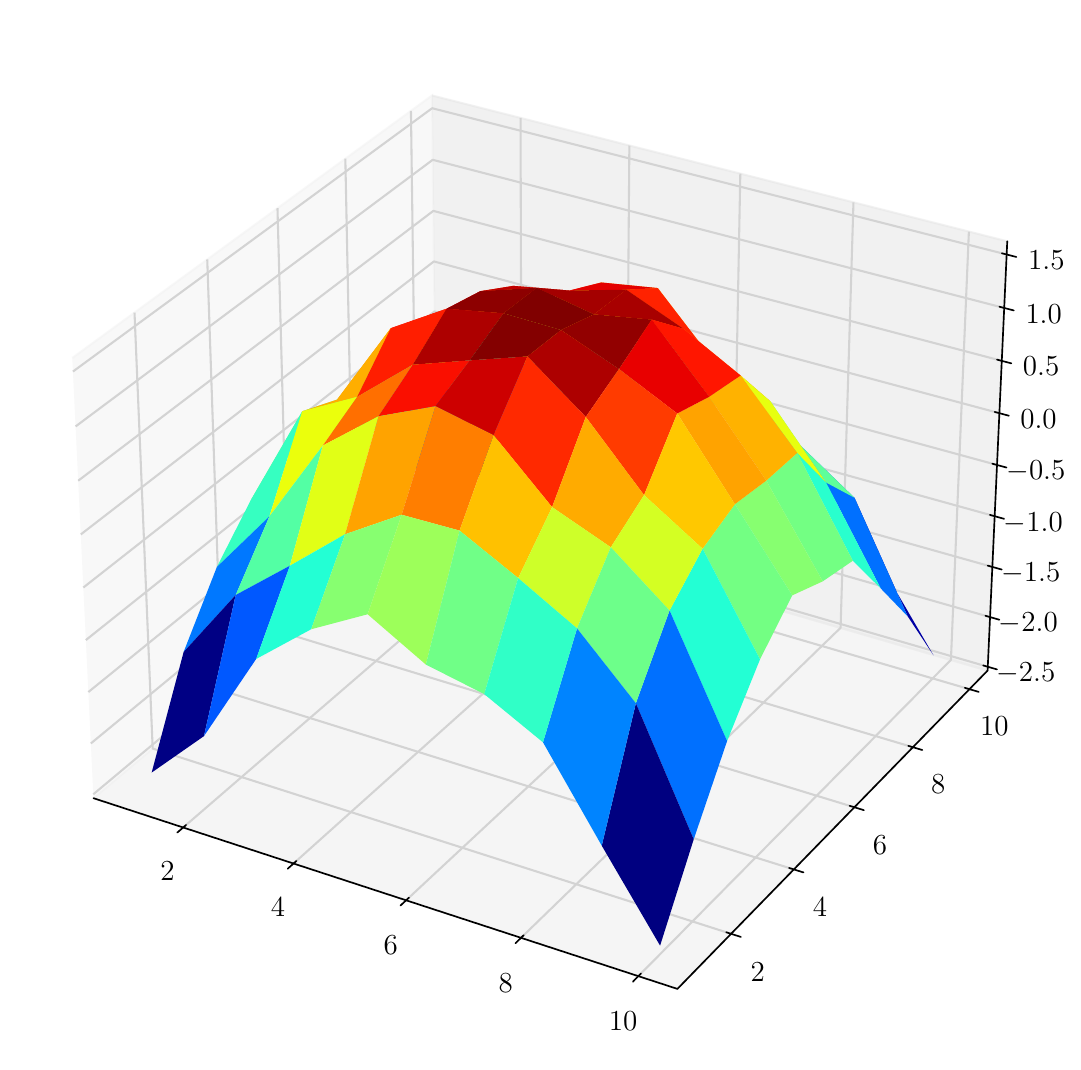}
     \end{subfigure}
     \hfill
     \begin{subfigure}[b]{0.24\textwidth}
         \centering
         \caption{$n=500$}
         \label{fig:normal-surface-500}
         \includegraphics[width=\textwidth]{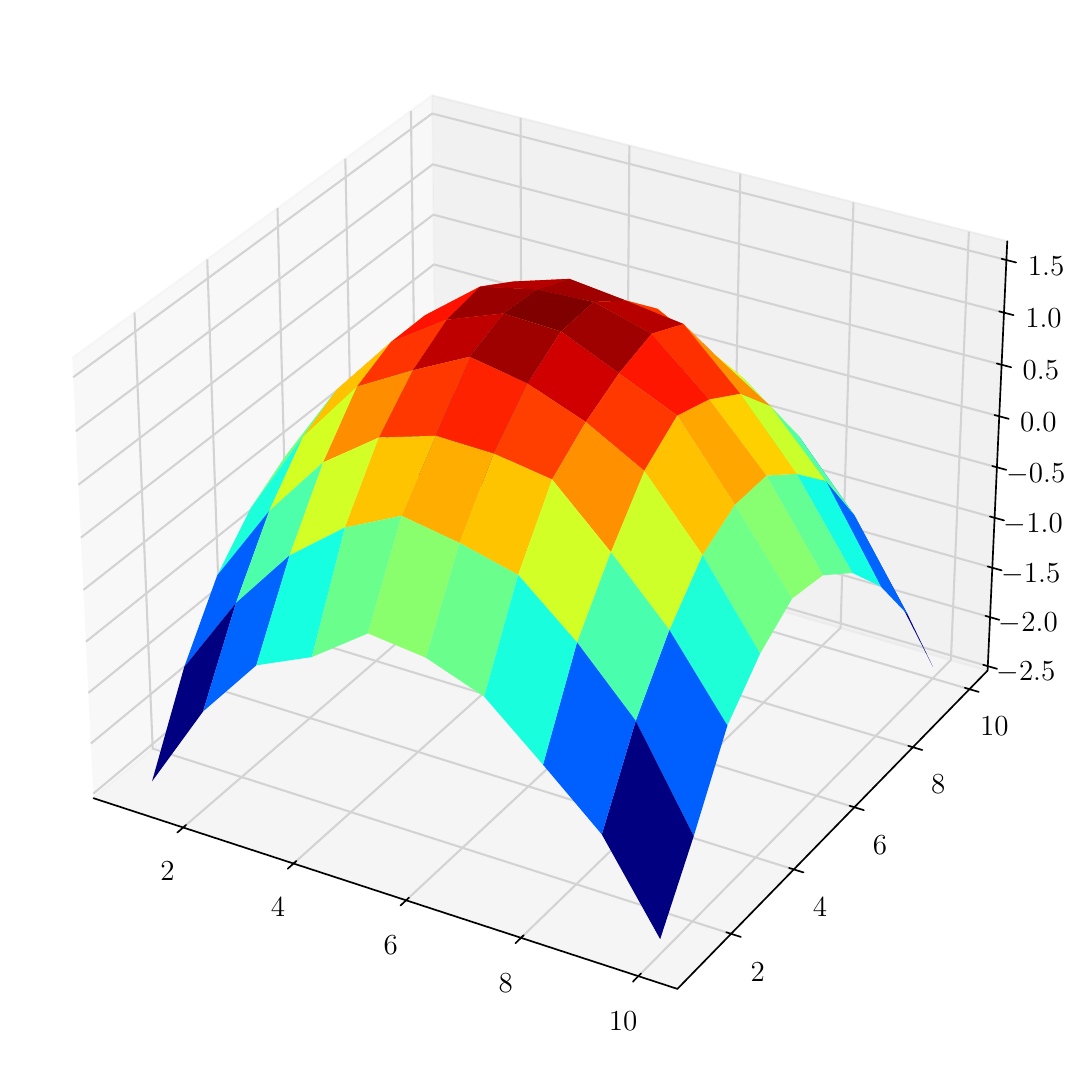}
     \end{subfigure}
\caption{Estimated joint Euclidean mirrors based on a sample of size $n$ for the simulation with Gaussian distributions in  Section~\ref{sec:normal-mirrors}.}
\label{fig:normal-surfaces}
\end{figure*}


\begin{figure}
    \centering
    \includegraphics[width=0.95\linewidth]{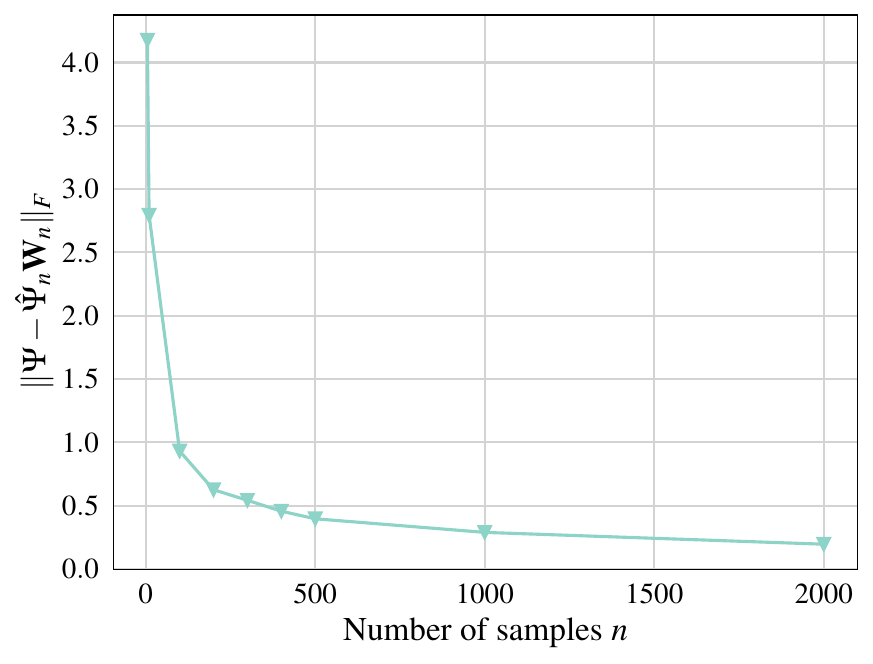}
    \caption{Error for the simulation described in Section \ref{sec:normal-mirrors} with increasing sample size $n$ from each observed distribution.}
    \label{fig:simulation-error}
\end{figure}

\subsection{Parameter recovery}
\label{sec:Simulated_recovery}

In order to demonstrate the consistency of the parameter estimation procedure described in Algorithm~\ref{alg:Recovery} and demonstrate the results outlined in Theorem~\ref{conj:consistent_param_estim}, we consider an example with simulated data similar to that presented in Section~\ref{sec:normal-mirrors}. 
In this example, we once again consider normal distributions and define $\mathcal{X} \subset \mathbb{R}^2$ to be the $[0,1] \times [0,1]$ plane. 
We construct a Euclidean $2$-realizable example by defining $F_x\sim\mathcal{N}(\mu_x,\sigma_x)$ for $x=(x_1,x_2)\in\mathcal{X}$, with mean and standard deviation of each distribution taking the form $\mu_x = 2(0.1+x_1)^2$ and $\sigma_x  = 2(0.1+x_2)^2$ for $i \in [m]$. 
In this case, the Wasserstein 2-distance takes the form $W_2(F_x, F_{x^\prime}) = [(\mu_x-\mu_{x^\prime})^2+(\sigma_{x} - \sigma_{x^\prime})^2]^{1/2}$. Hence, if $\mathcal{D}$ is set to the Wasserstein 2-distance, $(\mathcal F,\mathcal D)$ is Euclidean 2-realizable with mirror given by $f(x) = [2(0.1+x_1)^2, 2(0.1+x_2)^2]$.
As in the example presented in Section~\ref{sec:normal-mirrors} we set $m=100$ with observed points equally spaced along both 
dimensions of $\mathcal{X}$. 

In order the assess the effectiveness of the parameter recovery algorithm, we perform a leave-one-out analysis, where one parameter value $x\in\mathcal\{x_1,\dots,x_m\}$ is deleted, and an estimated mirror is generated from the remaining samples is used to recover the deleted label via Algorithm~\ref{alg:Recovery}. This procedure is performed in several experimental settings, where 
we vary the sample size $n$ obtained from each of the observed distributions while $m=100$ remains constant. The results of the simulation are presented in Figure~\ref{fig:recovery-sim}, where we see that as $n$ increases from $10$ to $10000$, the accuracy of Algorithm~\ref{alg:Recovery} converges to near perfect recovery, as prescribed by Theorem~\ref{conj:consistent_param_estim}. 

The number of observed distributions $m$ required to obtain accurate estimates will largely depend on the complexity of the function $f$ as well as the distribution of the observed parameters. For the example presented here, using $m=100$ appears to be sufficient to generate very accurate estimates when $n$ is large. On the other hand, the number of samples $n$ required to obtain good estimates will depend largely on the complexity of the underlying distributions. In this illustrative case, 
$n=1000$ generates a very accurate representation of each distribution and therefore lends itself to very accurate parameter estimates, 
but even values of $n$ as small as $n=10$ appear to be informative despite generating noisier estimates. 

\begin{figure}[t]
    \centering
    \begin{subfigure}{0.225\textwidth}
        \centering
        \caption{$n=10$}
        \includegraphics[width=\linewidth]{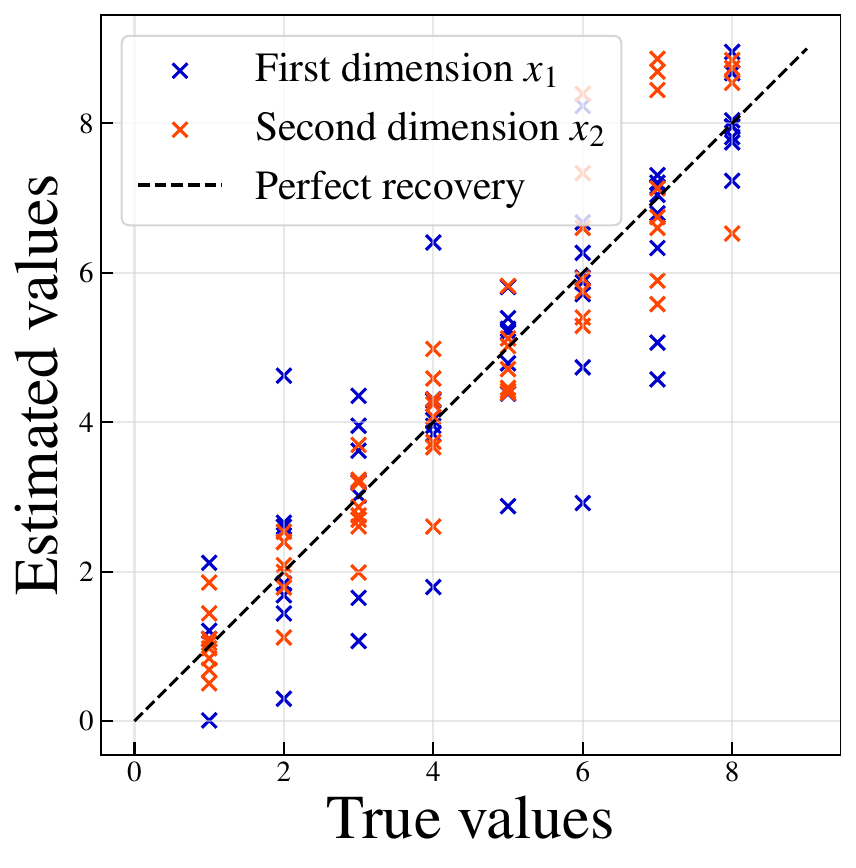}
    \end{subfigure}
    \hfill
    \begin{subfigure}{0.225\textwidth}
        \centering
        \caption{$n=100$}
        \includegraphics[width=\linewidth]{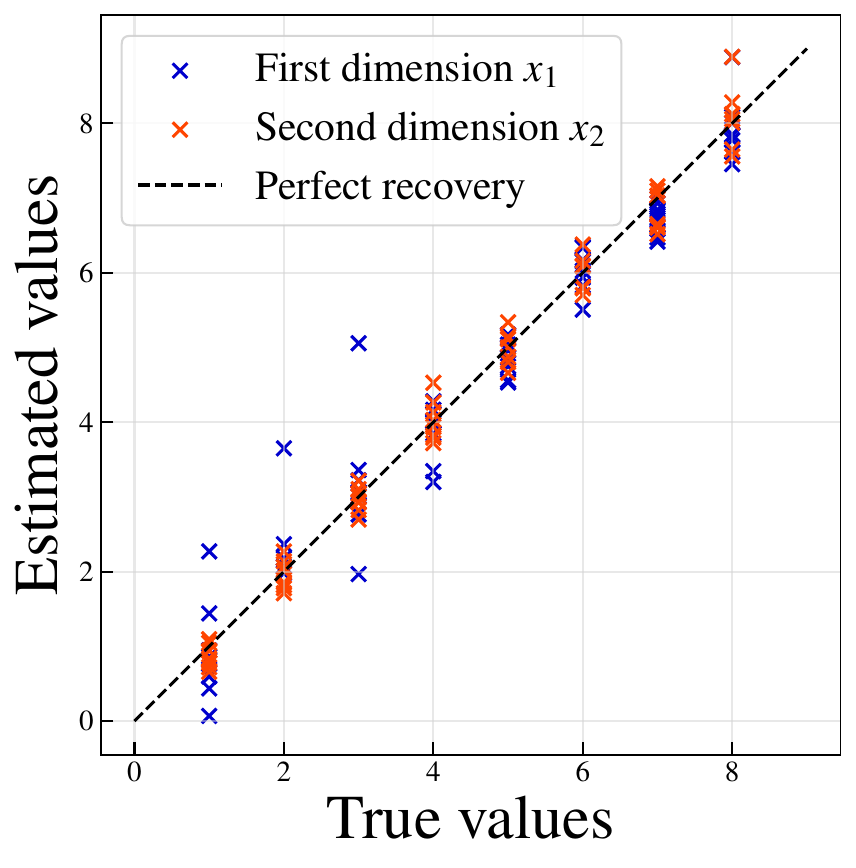}
    \end{subfigure}
    \begin{subfigure}{0.225\textwidth}
        \centering
        \caption{$n=1000$}
        \includegraphics[width=\linewidth]{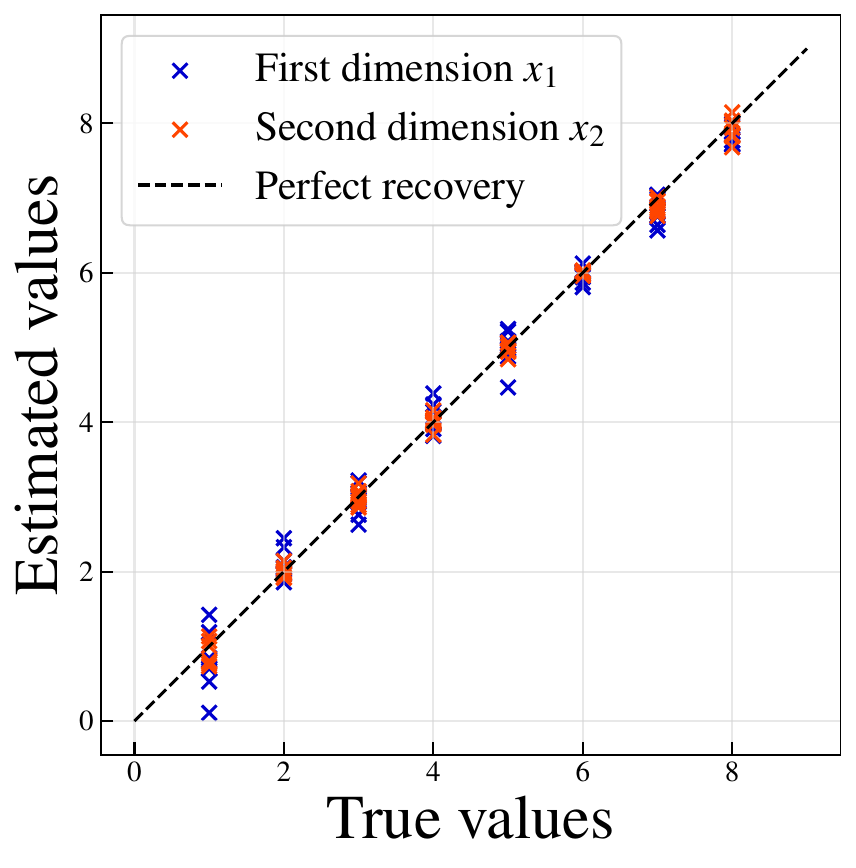}
    \end{subfigure}
    \hfill
    \begin{subfigure}{0.225\textwidth}
        \centering
        \caption{$n=10000$}
        \includegraphics[width=\linewidth]{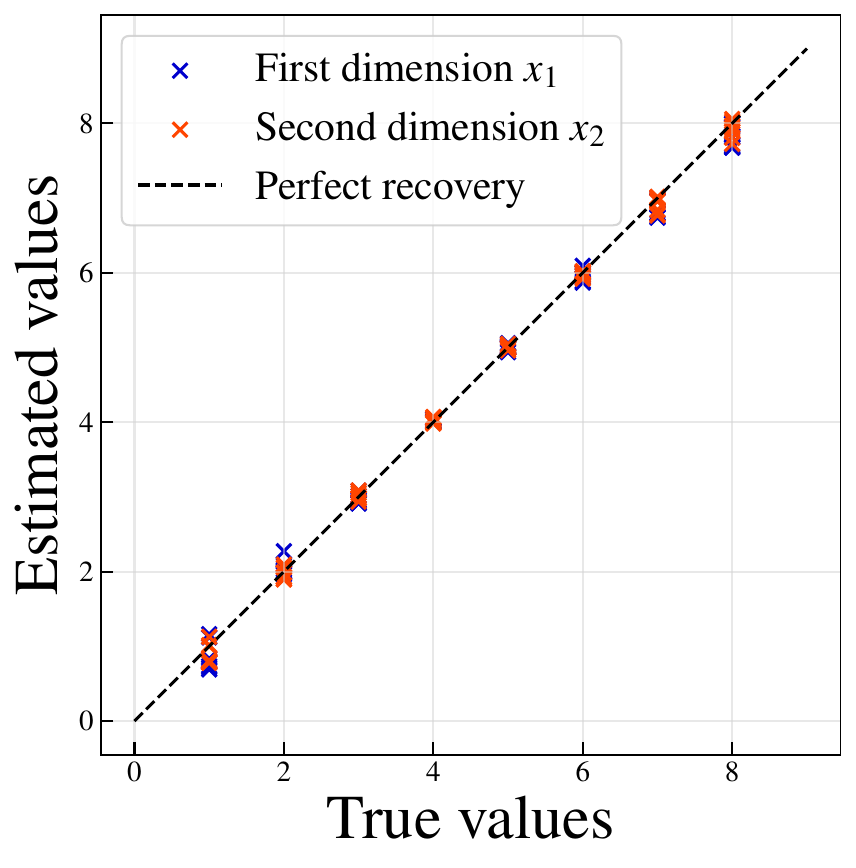}
    \end{subfigure}
    \caption{Parameter recovery performance with growing $n$ for the simulation with Gaussian distributions in Section~\ref{sec:Simulated_recovery}.}
    \label{fig:recovery-sim}
\end{figure}


\section{Application to LLM responses}
\label{sec:LLM-mirrors}

To demonstrate an application of Algorithm~\ref{Algo:JointMirror} to a setting in which the true mirror is unknown, we apply the mirror estimation procedure to a dataset of responses from large language models which were generated by querying LLMs with different temperature parameters with the prompt \textit{``Briefly describe R.A. Fisher’s work, in just two sentences, giving $w$\% weight to eugenics''}. In LLMs, the temperature is a parameter that can be adjusted to control the amount of randomness in the response generation process. Low temperature values tend to create more predictable responses while high temperatures allow for more randomness and produce more varied responses. Under this experimental setting, the sets of responses we generate are paramaterized by a two-dimensional parameter encoding both the temperature $t$ of the LLM as well as the weight parameter $w$ embedded within the prompt. We vary the weight parameter $w$ from $10\%$ to $90\%$ in increments of $10$ and the temperature $t$ from $0.1$ to $0.9$ in increments of $0.1$. We therefore observe data from a total of $m = 9 \times 9=81$ distinct parameter combinations, and for each of these combinations, we query the LLM $100$ times to generate $n=100$ samples. By embedding the resulting responses into the Euclidean space using \texttt{NomicEmbedv1.5} \citep{nussbaum2025nomic}, we can treat each response from the LLM as a sample from a distribution $F_x$ on the embedding space where $x \in \mathbb{R}^2$ encodes the LLM temperature $t$ and the prompt weight $w$. 

In this example, we take the Wasserstein 1-distance as our distance metric $\mathcal{D}$. Our mirror $f: \mathbb{R}^2 \to \mathbb{R}^c $ is therefore a function such that for all $x, \x^\prime \in \mathcal{X}$ the Euclidean distance between points of the mirror $\norm{f(x) - f(x^\prime)} $ provides a representation for the Wasserstein distance between the sets of embedded responses that were produced for different combinations of the weight $w$ and temperature $t$. 
In order to determine whether a set of distances is realizable by a low-dimensional manifold it is recommended to look for an elbow in the scree plot of the matrix of empirical pairwise distances $\boldsymbol{\hat \Delta}_{m,n}$ after double centering, plotted in Figure~\ref{fig:llm-scree}, following existing criteria for latent dimensionality selection \citep[see, for example][]{Zhu06}.
Figure~\ref{fig:llm-scree} provides empirical evidence that the assumption of exact Euclidean realizability, corresponding to isometric equivalence between $(\mathcal{F}, \mathcal{D})$ and $(\mathbb{R}^c, \mathcal{E})$, may be reasonable in this example: all eigenvalues of the estimated distance matrix are positive, 
implying that the distance matrix is exactly Euclidean. Moreover, the 
dimensionality selection criterion based on the scree plot \citep{Zhu06} suggests that the embedding dimension can be chosen to be 
relatively small ($c=5$) without losing substantial information.

\begin{figure}[t]
    \centering
    \includegraphics[width=0.95\linewidth]{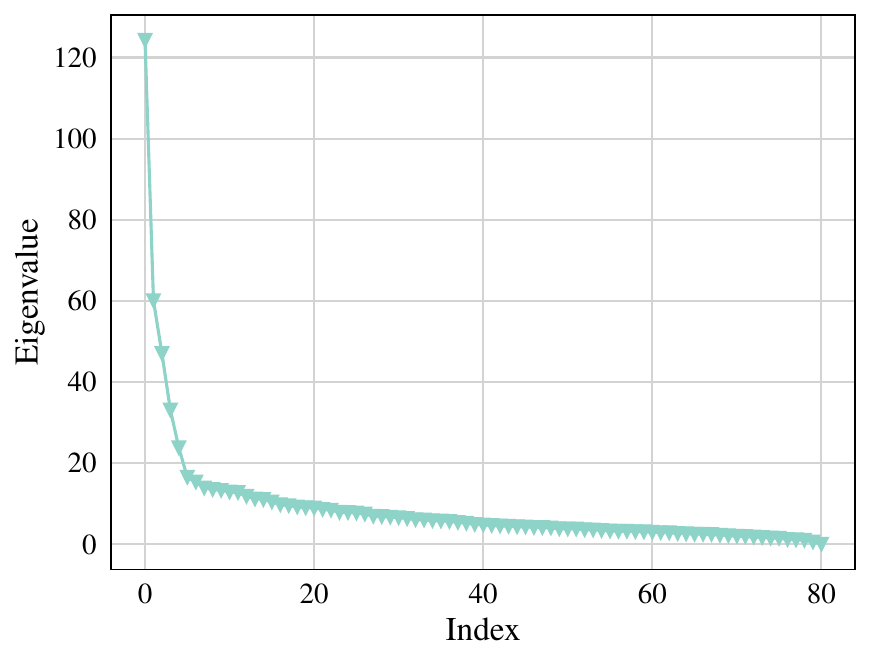}
    \caption{Eigenvalues of the doubly centered empirical distance matrix for the LLM application detailed in Section \ref{sec:LLM-mirrors}.}
    \label{fig:llm-scree}
\end{figure}

While the theory of Theorems~\ref{conj:joint-mirror} and~\ref{conj:consistent_param_estim} make use of linear interpolation via Delaunay triangulation, the choice of interpolation method used in Algorithm~\ref{Algo:JointMirror} is left open by design. Given the smoothness assumptions on the true mirror $f$ we find that both linear interpolation as well as B-splines \citep{eilers1996flexible}, which are smooth by construction, are a natural fit. To illustrate this choice, Figure~\ref{fig:surfaces} depicts the output of Algorithm \ref{Algo:JointMirror} for $c=1$ when using B-splines (Figure~\ref{fig:iso-surface-smooth}) as well as when using Delaunay linear interpolation (Figure~\ref{fig:iso-surface}). 

\begin{figure}[t!]
     \centering
     \begin{subfigure}[b]{0.425\textwidth}
         \centering
         \caption{Mirror constructed via Delaunay interpolation}
         \label{fig:iso-surface}
         \includegraphics[width=\textwidth]{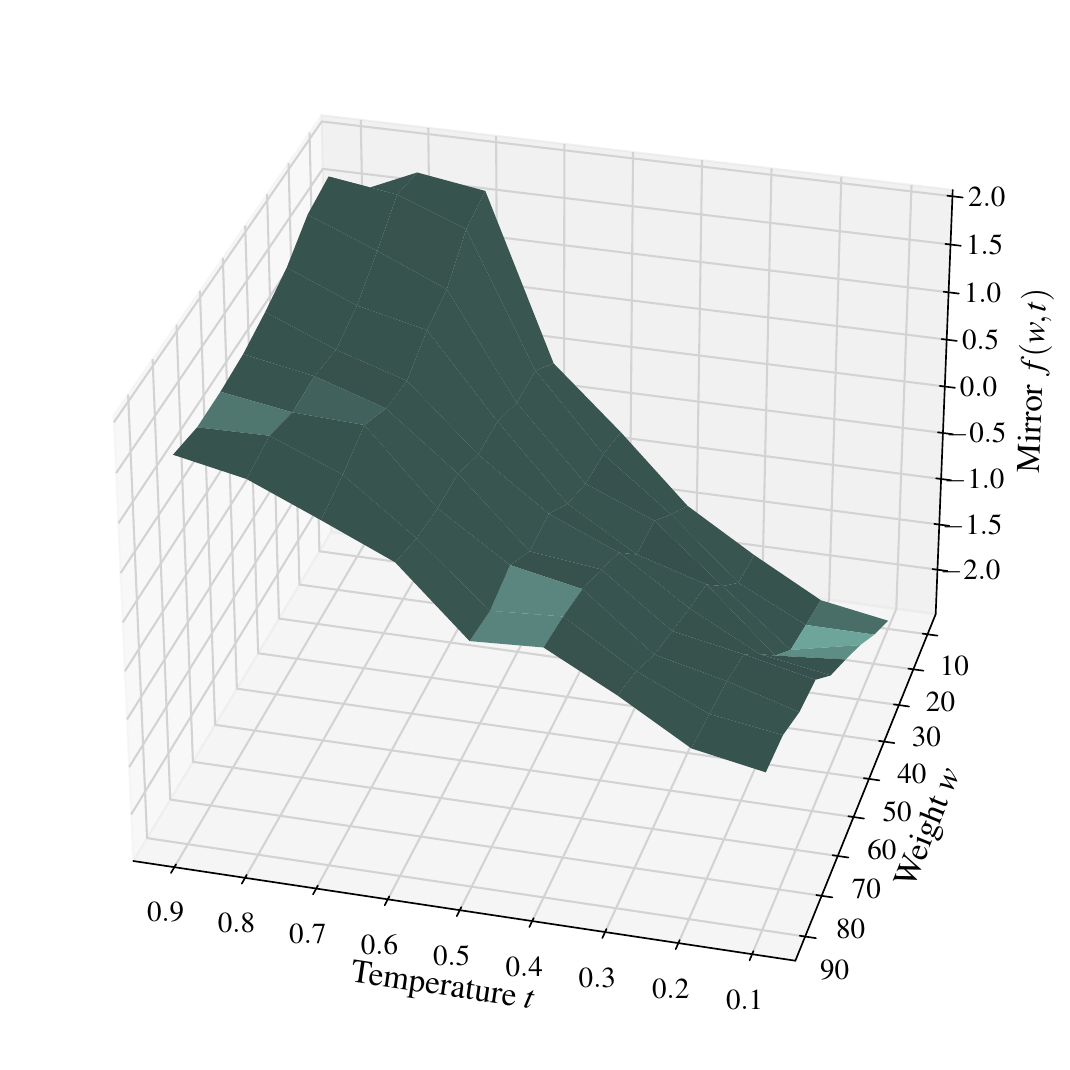}
     \end{subfigure}
     \hfill
     \begin{subfigure}[b]{0.425\textwidth}
         \centering
         \caption{Mirror constructed via B-splines}
         \label{fig:iso-surface-smooth}
         \includegraphics[width=\textwidth]{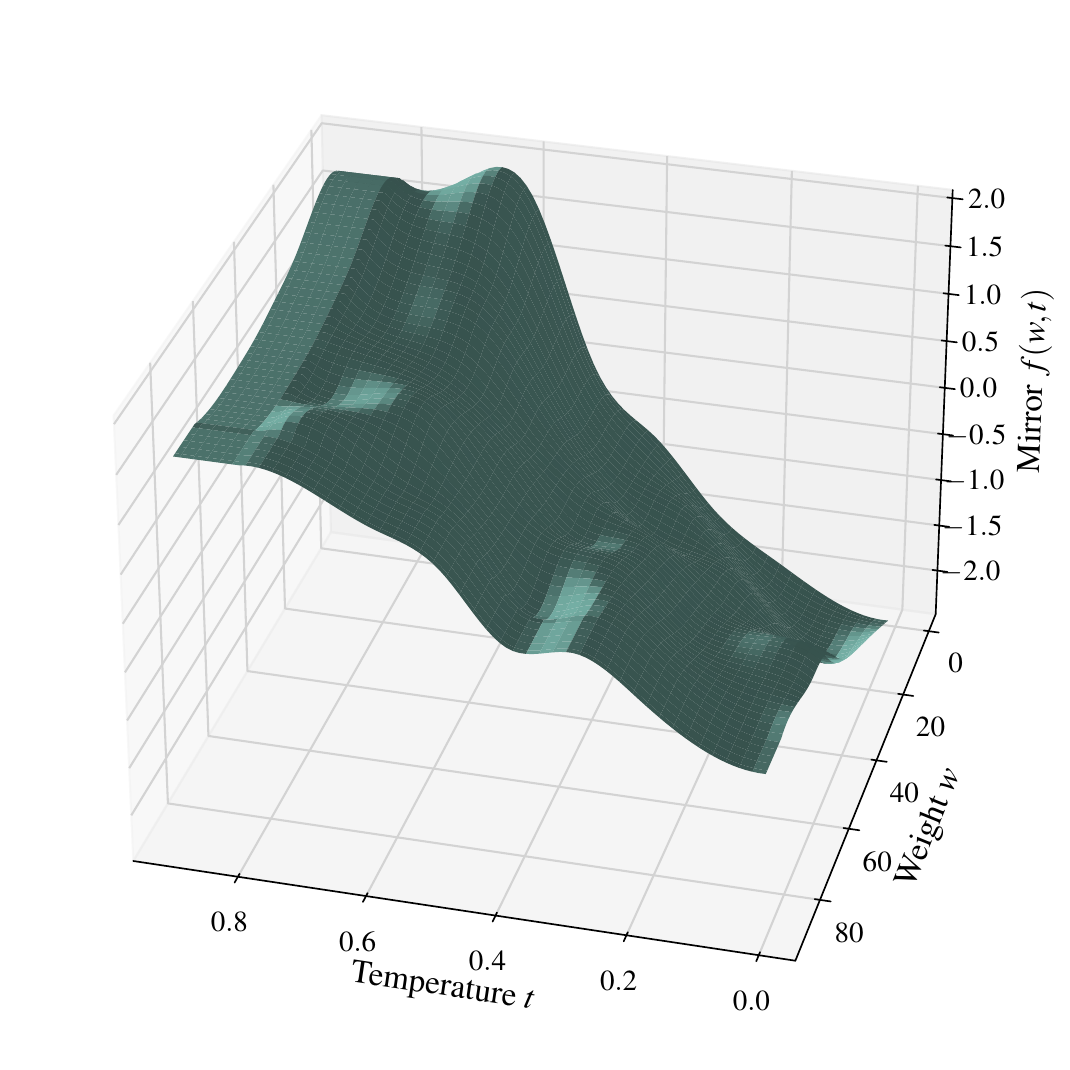}
     \end{subfigure}
        \caption{Estimated mirrors via (a) Delaunay interpolation and (b) B-splines for the application with LLMs in Section~\ref{sec:LLM-mirrors}.}
\label{fig:surfaces}
\end{figure}

\paragraph{Univariate joint Euclidean mirror.} In this example, as we do not have knowledge of the underlying distributions, we can not be certain that a mirror exists or that our sample size of $n=100$ is large enough to produce estimates of the mirror that provide an accurate representation of this structure. However, the estimated mirror does appear to show latent structure in the distribution of LLM responses. In particular, we see a well-ordered monotonic relationship between the estimated mirror and changes to both temperature and weight suggesting that changes to these parameters shift the distribution of responses in a consistent way that is measurable via the Wasserstein distance of the embedded responses. By examining the surfaces in Figure \ref{fig:surfaces}, we gain some understanding of the sensitivity of LLM responses to these changes in parameterization. For example, we see that the impact of changes in temperature is most pronounced between $0.4$ and $0.6$, and less consequential at the extreme values. In the setting where $\mathrm{dim}(\mathcal{X})>2$ or $c>1$ such an intuitive visualization of the mirror will not be possible. However, the full joint Euclidean mirror structure can still be effectively used for clustering of similar models as well as anomaly detection tasks. Meanwhile, marginal mirrors considering only one of the latent dimensions can be constructed from the full mirror to provide some visual intuition about the relationship between parameters. 
 
\begin{figure}[t]
     \centering
     \begin{subfigure}[b]{0.45\textwidth}
         \centering
         \caption{Joint mirror values, colored according to the temperature $t$}
         \label{fig:2d-mirrorT}
         \includegraphics[width=\textwidth]{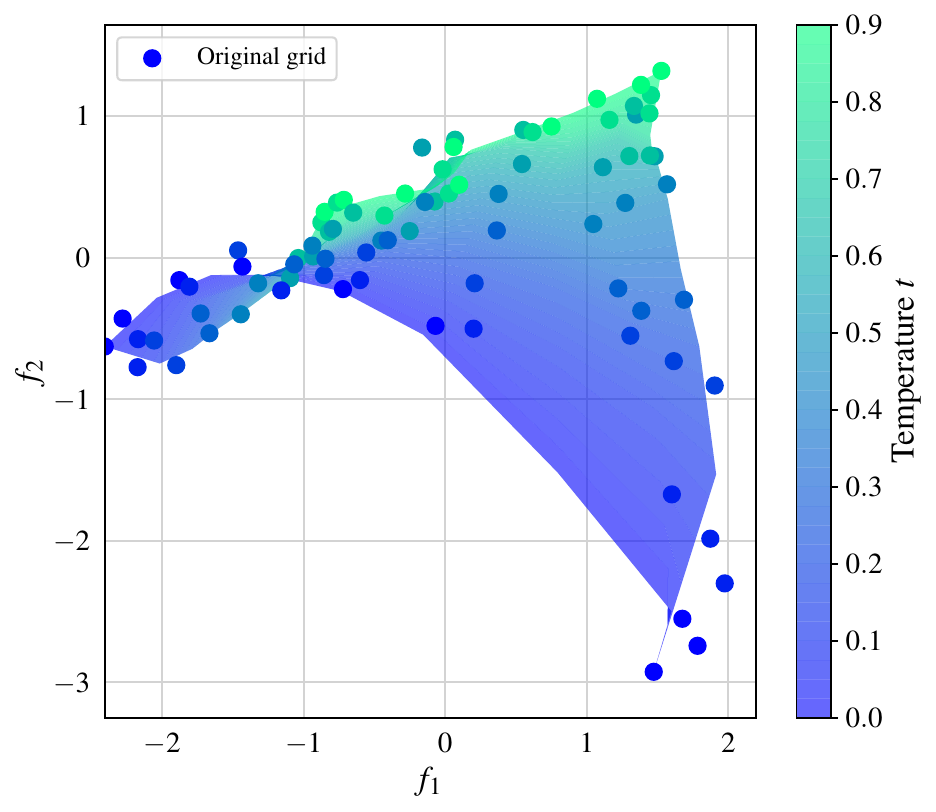}
         
     \end{subfigure}
     \hfill
     \begin{subfigure}[b]{0.45\textwidth}
         \centering
         \caption{Joint mirror values, colored according to the weight $w$}
         \label{fig:2d-mirrorW}
         \includegraphics[width=\textwidth]{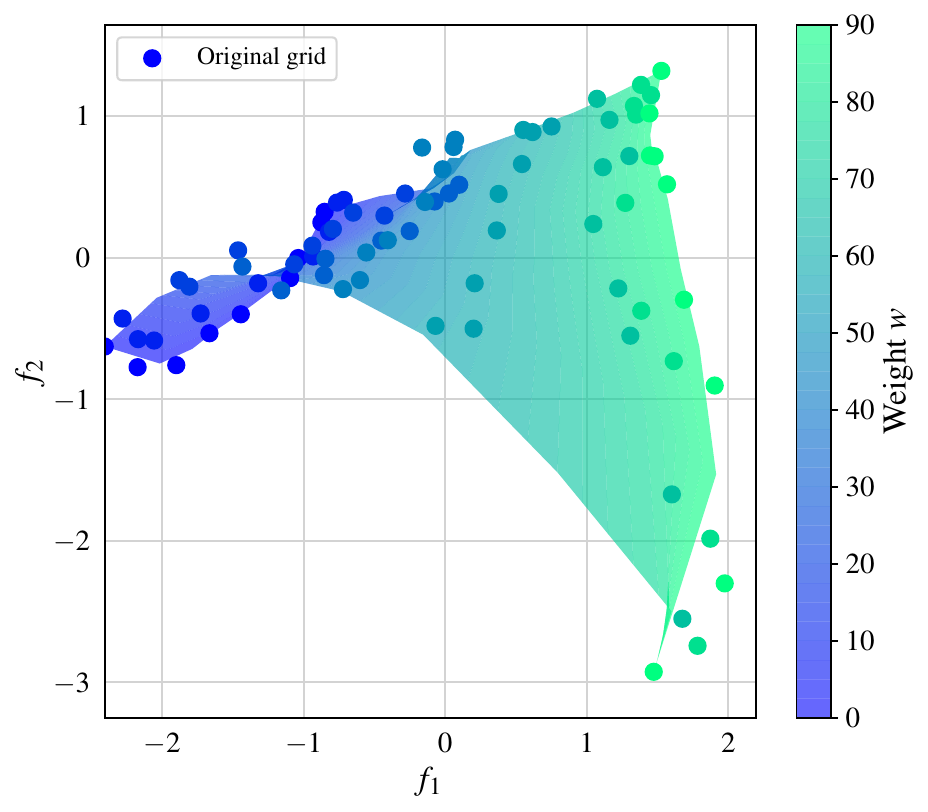}
         
     \end{subfigure}
        \caption{Scatterplot of the estimated two-dimensional mirror for the LLM example in Section \ref{sec:LLM-mirrors}, colored by the parameters.}
\label{fig:2dsurfaces}
\end{figure}

\paragraph{Multivariate joint Euclidean mirror.}
If we instead construct the mirror for the LLM case study 
using a mirror dimension of $c=2$, we obtain a collection of points in $\mathbb{R}^2$ encoding the distances between LLM responses, each corresponding to a particular weight and temperature pair. By interpolating these points we can construct a surface encoding the distance between every possible combination of parameters. In order to visualize this surface, we can color it such that each of the points with the same value of temperature are the same color or such that all points with a shared weight receive the same coloring. The plots depicting each of these colorings are presented in Figure~\ref{fig:2dsurfaces} and crucially, we see that the variation of weight and temperature across the surface occur on different axes. This suggests that the effects of varying temperature and weight shift the distribution of the outputs in ways that are roughly orthogonal to each other. 


\paragraph{Parameter recovery.} Given the fact that our parameters vary smoothly and monotonically over the embedding space and that each of them vary along distinct dimensions it follows that each point on the mirror surface must correspond to a distinct combination of parameters $w$ and $t$. Therefore, if we can place a collection of responses from an LLM on this surface, it should then be possible to approximately recover the weight and temperature used to generate these responses, via Algorithm~\ref{alg:Recovery}. 
In order to test this procedure we perform a type of leave-one-out validation using the LLM data, similarly to Section~\ref{sec:Simulated_recovery}. 
We delete the parameter labels from a single set of LLM responses and then make use of Algorithm~\ref{alg:Recovery} to recover these labels. 
The resulting estimates are plotted against the true parameter values in Figure~\ref{fig:recovery}. We see that while the estimates display a meaningful levels of both bias and variance, there is a robust linear relationship between the true values and the estimated values. From Theorem~\ref{conj:consistent_param_estim}, we expect that as both $m$ and $n$ increase, the accuracy of our parameter estimates will improve, assuming that the true mirror $f$ is well-behaved, 
as demonstrated on the simulated dataset in Section~\ref{sec:Simulated_recovery}.

\begin{figure}[t]
     \centering
     \includegraphics[width=0.45\textwidth]{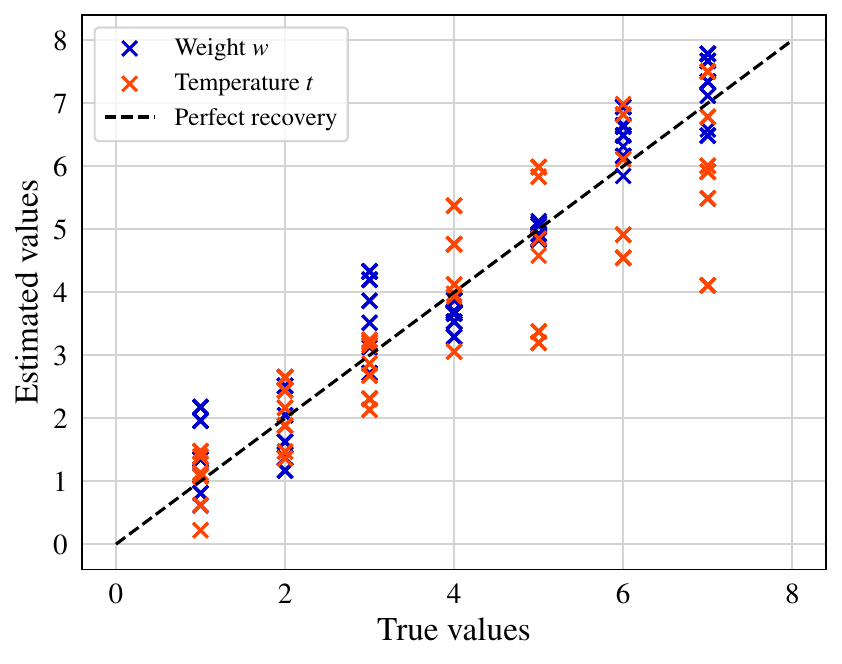}
    \caption{Performance of the leave-one-out parameter recovery procedure 
    on the LLM example described in Section \ref{sec:LLM-mirrors}.}
\label{fig:recovery}
\end{figure}

\section{Discussion}
\label{sec:Discussion}

The ability to recover latent structure in the distances between probability distributions is a concept with broad applications, including the representation of differences between distinct LLMs. In this work, we have introduced an algorithm to achieve this goal, which generalizes the concept of the Euclidean mirror, combining it with an underlying parametrization of probability distributions. Under the framework detailed in Sections~\ref{sec:mirror-background} and~\ref{Sec:Methodology}, we show that, when the distances between a set of distributions admits a low-dimensional representation, Algorithm~\ref{Algo:JointMirror} is able to produce a consistent estimate of this representation, given only a sample from a subset of these distributions. In a simulated dataset, where such a low-dimensional structure exists by design, we demonstrate that it can be successfully recovered (\textit{cf.} Section~\ref{sec:normal-mirrors}). When applying the methodology to an empirical dataset in Section~\ref{sec:LLM-mirrors}, we find evidence of latent structure in the Wasserstein distances between embedded responses from different LLMs, suggesting that this methodology could be effectively applied to understand which model attributes most significantly impact LLM output, which could possibly be used to predict the qualities of a model with attributes that we have not yet observed. We also demonstrate how the proposed procedure can be used to derive consistent estimates of model parameters from unlabeled samples, providing a method for \textit{parameter recovery}. In an application with LLMs, we show the presence of clear relationships between varying parameter values and dimensions of the latent space, which suggest that our algorithm provides an interesting and viable framework to analyze and compare the output of differently parameterized LLMs. 

\section*{Code}
Code to implement the methods proposed in this work, and reproduce the simulated experiments, is available in the Github repository \href{https://github.com/fraspass/llm_mirror} {\texttt{fraspass/llm\_mirror}}.

\section*{Acknowledgments}
This work was supported by the Engineering and Physical Sciences Research Council (EPSRC), grant number EP/Y002113/1, 
Office of Naval Research (ONR) Science of Autonomy award number N00014-24-1-2278,
Defense Advanced Research Projects Agency (DARPA) Artifical Intelligence Quantified award number HR00112520026.

\appendix

\section{Proofs of theoretical results}

\subsection{Proof of Proposition~\ref{prop:consistent-delta-estimation}} \label{proof:consistent-delta-estimation}

\begin{proof}
    
    Let $\mu\in\mathcal{F}$ and let $\mu_n$ denote its empirical counterpart. We apply Theorem~2 in \cite{fournier2015rate} to find that, for some constants $C_{\mu},c_{\mu}>0$ and $t \leq 1$, we have:
    \begin{equation}
    \mathbb{P}\left\{W(\mu, \mu_n) > t\right\} \leq C_\mu  \exp \left\{- c_\mu n t^{\max (q, 2)}\right\}.
    \end{equation}
    Substituting $t = \{\log^2(n)/n\}^{1/q}$ and using the moment condition yields
    \begin{equation}
    \mathbb{P}\left\{ W(\mu, \mu_n) > \left(\frac{\log^2(n)}{n}\right)^{1/q} \right\} \leq C \exp \left\{-c \log^2(n) \right\}. \label{eq:bound_W}
    \end{equation}
    for any choice of $\mu\in\mathcal{F}$, for constants $C,c>0$. 

    Now consider the individual terms of $\bDelta- \bDeltahat$, consisting of differences between the empirical and theoretical Wasserstein distances between distributions. Let $\mu$ and $\mu^\prime$ denote these distributions and let $\mu_n$ and $\mu_n^\prime$ be their empirical counterparts. Each entry of $\bDelta- \bDeltahat$ can be written as 
\begin{equation}
(\bDelta- \bDeltahat)_{i,j} = W(\mu, \mu^\prime) - W(\mu_n, \mu^\prime_n).
\end{equation}
    Define the quantities $A_n = \max\{W(\mu, \mu^\prime), W(\mu_n, \mu_n^\prime)\}$ and $B_n = \min\{W(\mu, \mu^\prime), W(\mu_n, \mu_n^\prime)\}$. By the triangle inequality, we get:
    \begin{equation}
        A_n \leq B_n + W(\mu, \mu_n) + W(\mu^\prime, \mu_n^\prime).
    \end{equation}
    Therefore: 
    \begin{equation}
    \vert (\bDelta- \bDeltahat)_{i,j} \vert = A_n-B_n \leq W(\mu, \mu_n) + W(\mu^\prime, \mu_n^\prime).
    \end{equation}    
    It follows that, if $\vert (\bDelta- \bDeltahat)_{i,j} \vert > 2t$ for some $t>0$, we must have either $W(\mu, \mu_n) > t$, or $W(\mu^\prime, \mu^\prime_n) > t$. Thus: 
    \begin{multline}
        \mathbb{P}\left\{ \vert (\bDelta- \bDeltahat)_{i,j} \vert > 2t  \right\} \\ \leq
            \mathbb{P}\left\{W(\mu, \mu_n) > t \right\} + \mathbb{P}\left\{W(\mu^\prime, \mu^\prime_n) > t \right\}. 
    \end{multline}
    Setting $t=\{\log^2(n)/n\}^{1/q}$ and using Equation~\eqref{eq:bound_W} gives:
    \begin{multline}
       \mathbb{P}\left\{ \vert (\bDelta- \bDeltahat)_{i,j} \vert > 2\left(\frac{\log^2(n)}{n}\right)^{1/q} \right\} \\ \le 2C \exp\{-c\log^2(n)\}.
    \end{multline}
    It follows that
    \begin{equation}
    \vert (\bDelta- \bDeltahat)_{i,j} \vert \leq \left(\frac{\log^2(n)}{n}\right)^{1/q}
    \label{eq:prop1_entrywise}
    \end{equation}
    uniformly with overwhelming probability for all $m^2$ elements of $\bDelta- \bDeltahat$. Therefore, we conclude that 
    \begin{equation}
    \| \bDelta - \bDeltahat\|_F \leq m\left(\frac{\log^2(n)}{n}\right)^{1/q}
    \end{equation}
    with overwhelming probability, 
    as desired. 
\end{proof}

\subsection{Proof of Theorem~\ref{conj:psi_convergence_growing_m_and_n}} \label{proof:psi_convergence_growing_m_and_n}

\begin{proof} 
Consider the matrices $\Bmat = -\frac{1}{2} \mathbf{H}_{m} \bDelta^{\odot 2}\mathbf{H}_m$ and $ \Bmathat = -\frac{1}{2} \mathbf{H}_{m} \bDeltahat^{\odot 2} \mathbf{H}_m$. 
By Proposition~\ref{prop:consistent-delta-estimation}, there is a value $\nu=\nu(m,n,q)$ such that with high probability for large $n$, $\|\bDelta-\bDeltahat\|_F\leq \nu.$ 
Let $\kappa=\lambda_1(\Bmat)/\lambda_c(\Bmat),$ and note that Assumptions~\ref{assum:boundedset}, \ref{assum:smallest-eigenvalue} imply that this has constant order. 
Set $\beta=2^{3/2}\nu/\lambda_c(\Bmat)$. The proof of Theorem 7 in \cite{athreya2025euclidean} shows that under such assumptions, there exists a $\Wmat\in\mathbb{O}(c)$ such that
\begin{multline}
\|\Psimat-\Psimathat\Wmat\|_F\\
\leq \beta\lambda_1^{1/2}(\Bmat)\left[2+4\beta\kappa^{1/2}+(1+2\beta)\beta/2^{3/2}\right].
\end{multline}
By Assumptions~\ref{assum:boundedset} and~\ref{assum:smallest-eigenvalue}, we have that $\lambda_1^{1/2}(\Bmat)=O(\sqrt{m})$, and $\kappa=O(1)$. 
Also, using Proposition~\ref{prop:consistent-delta-estimation} and Assumption~\ref{assum:smallest-eigenvalue}, we have that 
$\beta\leq C_{\beta}\{\log^2(n)/ n \}^{1/q}$ for a constant $C_\beta>0$, for sufficiently large $n$, with overwhelming probability. 
Therefore, there exists a constant $K>0$ such that:
\begin{equation}
    \|\Psimat-\Psimathat\Wmat\|_F \leq K 
    \sqrt{m}\left(\frac{\log^2(n)}{n}\right)^{1/q},
\end{equation}
for sufficiently large $n$, with overwhelming probability, as desired.
\end{proof}

\subsection{Proof of Theorem~\ref{conj:joint-mirror}} \label{proof:joint-mirror}
\begin{proof}
Let $\hat f_{m,n}$ be defined as the Delaunay interpolant in Equation~\eqref{eq:delauney_interpolant}. We proceed to bound the estimate error as
\begin{align}
    \hat f_{m,n}(x) &- f_{m,n}(x) = \sum_{j=1}^{d+1}\lambda_j(x) \tilde f_{m,n}\{v_j(\mathcal{P}_x)\} - f_{m,n}(x) \\ &= \sum_{j=1}^{d+1} \lambda_j(x)[\tilde f_{m,n}\{v_j(\mathcal{P}_x)\} - f_{m,n}\{v_j(\mathcal{P}_x)\}] \\ &\qquad + \left( \sum_{j=1}^{d+1} \lambda_j(x) f_{m,n}\{v_j(\mathcal{P}_x)\} - f_{m,n}(x) \right).
\end{align} 
Taking the norm of both sides and applying the triangle inequality yields
\begin{multline}    
 \norm{\hat f_{m,n}(x) - f_{m,n}(x)} \leq \\ \left\|\, \sum_{j=1}^{d+1} \lambda_j(x)[\tilde f_{m,n}\{v_j(\mathcal{P}_x)\} - f_{m,n}\{v_j(\mathcal{P}_x)\}] \,\right\| \\ + \left\|\, \sum_{j=1}^{d+1} \lambda_j(x) f_{m,n}\{v_j(\mathcal{P}_x)\} - f_{m,n}(x)\,\right\|.
\end{multline}
The first term corresponds to the 
error related to the measurement of the mirror $f_{m,n}$ at the observed parameter values, whereas the second term corresponds to the error resulting from interpolation. For the first term, we make use of Theorem~\ref{conj:psi_convergence_growing_m_and_n} under Assumption~\ref{cond:m-n-growth}, which gives that $\norm{\tilde f_{m,n}\{v_j(\mathcal{P}_x)\} - f_{m,n}\{v_j(\mathcal{P}_x)\}} \leq \delta_{m,n}$ with $\delta_{m,n} \to 0$. Hence: 
\begin{multline}
\sum_{j=1}^{d+1} \lambda_j(x) \norm{\tilde f_{m,n}\{v_j(\mathcal{P}_x)\} - f_{m,n}\{v_j(\mathcal{P}_x)\}} \leq \\ \delta_{m,n} \sum_{j=1}^{d+1} \lambda_j(x) = \delta_{m,n}.
\end{multline}
For the second term, we make use of the $C$-Lipschitz continuity of $f_{m,n}$ to write
\begin{multline}
    \left\|\, \sum_{j=1}^{d+1} \lambda_j(x) f_{m,n}\{v_j(\mathcal{P}_x)\} - f_{m,n}(x) \,\right\| \leq \\ \sum_{j=1}^{d+1} \lambda_j(x) \norm{f_{m,n}\{v_j(\mathcal{P}_x)\} - f_{m,n}(x)} \leq C\varepsilon_m
\end{multline}
where $\varepsilon_m$ converges to zero by Assumption~\ref{cond:dense-m}. In particular, the assumption ensures that there exists a sequence $\varepsilon_1,\varepsilon_2,\dots$ with $\varepsilon_m \to 0$, such that $\mathrm{diam}(\mathcal{P}_\ell) \leq \varepsilon_m$ for every simplex $\mathcal{P}_\ell$ in the Delaunay triangulation of $\{x_1,\dots,x_m\}$, where $\mathrm{diam}(S) = \max_{x,x^\prime\in S}\norm{x-x^\prime}$ is the diameter of the simplex $S$. 
Combining these two bounds, we find that for all $x \in \mathcal{X}_m$: 
\begin{equation}
    \norm{\hat f_{m,n}(x) - f_{m,n}(x)} \leq \delta_{m,n} + C\varepsilon_m \to 0. 
\end{equation}
Therefore, with high probability
\begin{equation}
\sup_{x \in \mathcal{X}_m} \norm{\hat f_{m,n}(x) - f_{m,n}(x)} \to 0,
\end{equation}
which proves the result. 
\end{proof}

\subsection{Proof of Theorem~\ref{conj:consistent_param_estim}} \label{proof:consistent_param_estim}

\begin{proof}
    Recall that the parameter estimate $\hatx$ is defined as 
    \begin{equation}
        \argmin_{s\in\mathcal{X}_m} \norm{\hat f(s)- \hat f^\ast},
    \end{equation}
    where $\hat{f}^*$ is an estimate for $f(x^*)$ based on the unlabeled mirror value, and is \emph{not} the same as $\hat{f}(x^*)$, which is the evaluation of the estimated mirror at the exact (unknown) value $x^*$. In most cases we expect an exact solution to this minimization procedure to exist such that $\hat f (\hatx) = \hat f^\ast$; however, to cover the general case, we denote the difference by $\eta\in\mathbb{R}^c$, such that $\hat f (\hatx) = \hat f^\ast + \eta$. We begin by writing 
    \begin{equation}
    \hatx - x^\ast = f^{-1}\{f(\hatx)\} - f^{-1}\{f(x^{\ast})\}.
    \end{equation}
    Let $\mathcal{B}_{\delta}(x^\ast)$ denote the ball of radius $\delta$ centered at $x^\ast$ where $\delta = \norm{x^\ast - \hatx}$. The mean value theorem guarantees the existence of $\xi \in \mathcal{B}_{\delta}(x^\ast)$ such that 
    \begin{multline}
    \norm{f^{-1}\{f(\hatx)\} - f^{-1}\{f(x^{\ast})\}} \leq \\ \norm{\mathcal{J}_{f^{-1}}(\xi)[f(\hatx ) - f(x^\ast)]}.
    \end{multline}
    It follows that,
    \begin{align}
    \|&\, \hatx - x^\ast\| \leq \norm{\mathcal{J}_{f^{-1}}(\xi)[f(\hatx ) - f(x^\ast)]} \\ &= \norm{\mathcal{J}_{f^{-1}}(\xi)[f(\hatx ) - \hat f(\hatx) + \hat f(\hatx) - f(x^\ast)]} \\ &= \norm{\mathcal{J}^{-1}_{f}(\xi)[f(\hatx ) - \hat f(\hatx) + \hat f^\ast - f(x^\ast) + \eta]} \\ &\leq \frac{\norm{f(\hatx ) - \hat f(\hatx)}}{a} + \frac{\norm{\hat f^\ast - f(x^\ast)}}{a} + \frac{\norm{\eta}}{a},
    \end{align}
    where $a$ is a uniform lower bound on the smallest singular value of the Jacobian $\mathcal{J}_f(\cdot)$. 
    The first of these terms converges to zero as $m,n \to \infty$ by Theorem~\ref{conj:joint-mirror}, and the second converges to zero by Theorem~\ref{conj:psi_convergence_growing_m_and_n}. In order to bound $\norm{\eta}$, we note that 
    \begin{align}
        \argmin_{s\in \mathcal{X}_m} &\norm{\hat f(s)- \hat f^\ast} \leq \norm{\hat f(x^\ast) - \hat f^\ast} \\ &\leq \norm{\hat f(x^\ast) - f(x^\ast)} + \norm{f(x^\ast) - \hat f^\ast}.
    \end{align}
    Once again, the first term converges to zero by Theorem~\ref{conj:joint-mirror} and the second term converges to zero by Theorem~\ref{conj:psi_convergence_growing_m_and_n}. 
    Hence: 
    \begin{equation}
    \norm{\hatx - x^\ast} 
    \rightarrow 0
    \end{equation}
    for $m,n\to\infty$ with high probability, which gives the desired result. 
\end{proof}


\begin{thebibliography}{22}
\expandafter\ifx\csname natexlab\endcsname\relax\def\natexlab#1{#1}\fi
\expandafter\ifx\csname url\endcsname\relax
  \def\url#1{\texttt{#1}}\fi
\expandafter\ifx\csname urlprefix\endcsname\relax\def\urlprefix{URL: }\fi

\bibitem[{Acharyya et~al.(2025)Acharyya, Agterberg, Park and Priebe}]{acharyya2025concentration}
Acharyya, A., Agterberg, J., Park, Y., and Priebe, C.~E. (2025) Concentration bounds on response-based vector embeddings of black-box generative models.
\newblock \textit{arXiv preprint arXiv:2511.08307}.

\bibitem[{Acharyya et~al.(2024)Acharyya, Trosset, Priebe and Helm}]{acharyya2024consistent}
Acharyya, A., Trosset, M.~W., Priebe, C.~E., and Helm, H.~S. (2024) Consistent estimation of generative model representations in the data kernel perspective space.
\newblock \textit{arXiv preprint arXiv:2409.17308}.

\bibitem[{Athreya et~al.(2025)Athreya, Lubberts, Park and Priebe}]{athreya2025euclidean}
Athreya, A., Lubberts, Z., Park, Y., and Priebe, C. (2025) {Euclidean Mirrors and Dynamics in Network Time Series}.
\newblock \textit{Journal of the American Statistical Association}, \textbf{120}, 1025--1036.

\bibitem[{de~Boor(1962)}]{deBoor62}
de~Boor, C. (1962) Bicubic spline interpolation.
\newblock \textit{Journal of Mathematics and Physics}, \textbf{41}, 212--218.

\bibitem[{Chen and Xu(2004)}]{Chen04}
Chen, L. and Xu, J.-C. (2004) Optimal {Delaunay} triangulations.
\newblock \textit{Journal of Computational Mathematics}, \textbf{22}, 299--308.

\bibitem[{Chewi et~al.(2025)Chewi, Niles-Weed and Rigollet}]{chewi2025statistical}
Chewi, S., Niles-Weed, J., and Rigollet, P. (2025) \textit{Statistical optimal transport}.
\newblock Springer.

\bibitem[{Eilers and Marx(1996)}]{eilers1996flexible}
Eilers, P.~H. and Marx, B.~D. (1996) {Flexible smoothing with B-splines and penalties}.
\newblock \textit{Statistical Science}, \textbf{11}, 89--121.

\bibitem[{Fournier and Guillin(2015)}]{fournier2015rate}
Fournier, N. and Guillin, A. (2015) On the rate of convergence in {Wasserstein} distance of the empirical measure.
\newblock \textit{Probability Theory and Related Fields}, \textbf{162}, 707--738.

\bibitem[{Gillette and Kur(2024)}]{Gillette24}
Gillette, A. and Kur, E. (2024) Algorithm 1049: The delaunay density diagnostic.
\newblock \textit{ACM Transactions on Mathematical Software}, \textbf{50}.

\bibitem[{Helm et~al.(2025)Helm, Acharyya, Park, Duderstadt and Priebe}]{Helm2025}
Helm, H., Acharyya, A., Park, Y., Duderstadt, B., and Priebe, C. (2025) Statistical inference on black-box generative models in the data kernel perspective space.
\newblock In \textit{Findings of the Association for Computational Linguistics: ACL 2025} (eds. W.~Che, J.~Nabende, E.~Shutova and M.~T. Pilehvar), 3955--3970. Association for Computational Linguistics.

\bibitem[{Kahng et~al.(2025)Kahng, Tenney, Pushkarna, Liu, Wexler, Reif, Kallarackal, Chang, Terry and Dixon}]{Kahng25}
Kahng, M., Tenney, I., Pushkarna, M., et~al. (2025) Llm comparator: Interactive analysis of side-by-side evaluation of large language models.
\newblock \textit{IEEE Transactions on Visualization and Computer Graphics}, \textbf{31}, 503--513.

\bibitem[{Kruskal(1964)}]{kruskal1964multidimensional}
Kruskal, J.~B. (1964) Multidimensional scaling by optimizing goodness of fit to a nonmetric hypothesis.
\newblock \textit{Psychometrika}, \textbf{29}, 1--27.

\bibitem[{Nussbaum et~al.(2025)Nussbaum, Morris, Mulyar and Duderstadt}]{nussbaum2025nomic}
Nussbaum, Z., Morris, J.~X., Mulyar, A., and Duderstadt, B. (2025) {Nomic Embed: Training a Reproducible Long Context Text Embedder}.
\newblock \textit{Transactions on Machine Learning Research}.

\bibitem[{Panaretos and Zemel(2019)}]{panaretos2019statistical}
Panaretos, V.~M. and Zemel, Y. (2019) Statistical aspects of {Wasserstein} distances.
\newblock \textit{Annual Review of Statistics and Its Application}, \textbf{6}, 405--431.

\bibitem[{Schumaker(2007)}]{Schumaker_2007}
Schumaker, L. (2007) \textit{Spline Functions: Basic Theory}.
\newblock Cambridge Mathematical Library. Cambridge University Press.

\bibitem[{Stone(1994)}]{Stone94}
Stone, C.~J. (1994) The use of polynomial splines and their tensor products in multivariate function estimation.
\newblock \textit{The Annals of Statistics}, \textbf{22}, 118--171.

\bibitem[{Tan et~al.(2024)Tan, Zeng, Tian, Liu, Yin and Jiang}]{tan2024}
Tan, Z., Zeng, Q., Tian, Y., et~al. (2024) Democratizing large language models via personalized parameter-efficient fine-tuning.
\newblock In \textit{Proceedings of the 2024 Conference on Empirical Methods in Natural Language Processing} (eds. Y.~Al-Onaizan, M.~Bansal and Y.-N. Chen), 6476--6491. Association for Computational Linguistics.

\bibitem[{Tao and Vu(2010)}]{tao2010random}
Tao, T. and Vu, V. (2010) Random matrices: Universality of local eigenvalue statistics up to the edge.
\newblock \textit{Communications in Mathematical Physics}, \textbf{298}, 549--572.

\bibitem[{Woźniak et~al.(2024)Woźniak, Koptyra, Janz, Kazienko and Kocoń}]{wozniak24}
Woźniak, S., Koptyra, B., Janz, A., Kazienko, P., and Kocoń, J. (2024) Personalized large language models.
\newblock In \textit{2024 IEEE International Conference on Data Mining Workshops (ICDMW)}, 511--520.

\bibitem[{Zhang et~al.(2025)Zhang, Rossi, Kveton, Shao, Yang, Zamani, Dernoncourt, Barrow, Yu, Kim, Zhang, Gu, Derr, Chen, Wu, Chen, Wang, Mitra, Lipka, Ahmed and Wang}]{zhang2025personalization}
Zhang, Z., Rossi, R.~A., Kveton, B., et~al. (2025) Personalization of large language models: A survey.
\newblock \textit{Transactions on Machine Learning Research}.

\bibitem[{Zheng et~al.(2023)Zheng, Chiang, Sheng, Zhuang, Wu, Zhuang, Lin, Li, Li, Xing, Zhang, Gonzalez and Stoica}]{Zheng23}
Zheng, L., Chiang, W.-L., Sheng, Y., et~al. (2023) {Judging LLM-as-a-judge with MT-bench and Chatbot Arena}.
\newblock In \textit{Proceedings of the 37th International Conference on Neural Information Processing Systems}. Curran Associates Inc.

\bibitem[{Zhu and Ghodsi(2006)}]{Zhu06}
Zhu, M. and Ghodsi, A. (2006) Automatic dimensionality selection from the scree plot via the use of profile likelihood.
\newblock \textit{Computational Statistics \& Data Analysis}, \textbf{51}, 918--930.

\end{thebibliography}
\end{document}